\documentclass[sigconf,authorversion,nonacm]{aamas} 

\usepackage{amsmath}
\usepackage{mathtools}
\newtheorem{thm}{Theorem}
\newtheorem{lemma}{Lemma}
\newtheorem{proposition}{Proposition}
\newtheorem{definition}{Definition}
\newtheorem{example}{Example}
\newtheorem{corollary}{Corollary}

\usepackage{mdframed}
\usepackage[mathscr]{eucal}
\usepackage{accents}

\DeclareMathOperator{\argmax}{argmax}

\usepackage{balance} 

\pdfoutput=1

\title[A Language for Representing Combinatorial Auctions]{
A General Framework for the Logical Representation of Combinatorial Exchange Protocols 
}

\author{Munyque Mittelmann}
\affiliation{
  \institution{Universit\'e de Toulouse - IRIT}
  \city{Toulouse, France}}
\email{munyque.mittelmann@irit.fr}

\author{Sylvain Bouveret}
\affiliation{
  \institution{Universit\'e Grenoble Alpes - LIG}
  \city{Grenoble, France}}
\email{sylvain.bouveret@imag.fr}

\author{Laurent Perrussel}
\affiliation{
  \institution{Universit\'e de Toulouse - IRIT}
  \city{Toulouse, France}}
\email{laurent.perrussel@irit.fr}

\begin{abstract}
The goal of this paper is to propose a framework for representing and reasoning about the rules governing a combinatorial exchange. Such a framework is at first interest as long as we want to build up digital marketplaces based on auction, a widely used mechanism for automated transactions. Combinatorial exchange is the most general case of auctions, mixing the double and combinatorial variants: agents bid to trade bundles of goods. Hence the framework should fulfill two requirements: (i) it should enable bidders to express their bids on combinations of goods and (ii) it should allow describing the rules governing some market, namely the legal bids, the allocation and payment rules. To do so, we define a logical language in the spirit of the \textit{Game Description Language}: the \textit{Combinatorial Exchange Description Language} 
is the first language for describing combinatorial exchange in a logical framework. The contribution is two-fold: first, we illustrate the general dimension by representing different kinds of protocols, and second, we show how to reason about auction properties in this machine-processable language.
\end{abstract}


\newcommand{\BibTeX}{\rm B\kern-.05em{\sc i\kern-.025em b}\kern-.08em\TeX}

\begin{document}


\pagestyle{fancy}
\fancyhead{}


\maketitle 

\section{Introduction}
Auction-based markets are widely used for automated business transactions. There are numerous variants whether we consider single or multiple goods, single or multiple units, single or double-side \cite{auction1999klemperer,Krishna}. For a fixed set of parameters, the auction protocol, i.e.,\ the bidding, allocation, and payment rules, may also differ. Building  intelligent agents that can  switch  between different auctions and process their rules is a key issue for building automated auction-based marketplaces. To do so, auctioneers should at first describe the rules governing an auction and second allow bidders to express complex bids. The aim of this paper is to propose such language with  clear semantics  enabling us to derive properties. Hereafter, we consider \emph{combinatorial exchanges}  which are the most general case for auctions, mixing double and combinatorial variants
\cite{lubin2008ice}.

More precisely, such \emph{Combinatorial Exchange Description Language} should address the following dimensions:
\begin{description}
    \item[Agent] (i) one seller and multiple buyers, or vice-versa (single-side auctions); (ii) multiple sellers and buyers (double-side auctions); (iii) multiple bidders that can be both sellers and buyers (i.e., agents are traders);
    \item[Unit]  (i) single-unit or (ii) multi-unit auction;
    \item[Good type] (i) single-good or (ii) multiple-goods; for that latter case a bid can 
    consider units and operators such as ``1 table \emph{and} 4 chairs''.
    \item[Bidding protocol] the auction protocol may be (i) one-shot (e.g.,\ sealed-bid auction) or (ii) iterative  (e.g.,\ ascending and descending auctions);
   \item[Allocation and payment] the protocol should (i) detail how the winners are determined \cite{Xia2005}, 
   (ii) quantify the money transfer (e.g., first price or Vickrey–Clarke–Groves payment \cite{Parkes2005}).
\end{description}

In the spirit of the \emph{General Game Playing} \cite{GT14} where games are described with the help of a logical language, namely the \emph{Game Description Language} (GDL), we propose the \emph{Combinatorial Exchange Description Language} (CEDL) which is based on the Auction Description Language (ADL) \cite{MP2020}.   
CEDL includes a bidding language that can represent a wide range of auctions from a single-side, single-unit and good auction (as a single-unit Vickrey Auction) to Iterative Combinatorial Exchange \cite{Parkes2005}. As for GDL and ADL, we propose a precise semantics based on state-transition models, that gives a clear meaning to the properties describing an auction.
CEDL embeds the Tree-Based Bidding Language \cite{Parkes2005}, which generalizes known languages such as XOR/OR \cite{Nisan2000} to combinatorial exchange, where agents should be able to express preferences for both buying and selling goods.
To the best of our knowledge, CEDL is the first framework offering a unified perspective on an auction mechanism: (i) representation on how to bid and (ii) representation of the protocol including allocation and payment. Such a framework offers two benefits: (i) with this language, one can represent many kinds of auctions in a compact way and (ii) the precise state-transition semantics can be used to derive key properties.

CEDL extends ADL in numerous ways: ADL only focuses on one-side auctions and can not represent double auctions; next, ADL is also unable to consider the multiple-good dimension. These two dimensions are now considered in CEDL and thus allow (i) to handle the problem of re-allocating goods, as in an exchange, and (ii) to embed a bidding language to address combinatorial markets.

The paper is organized as follows: in Section~\ref{sec:preliminaries} we  review the key components for describing auctions; we next detail in Section~\ref{sec:tbbl} the bidding language that will be embedded in CEDL and exhibit its key properties. Section~\ref{sec:cedl} details CEDL: semantics and syntax. Section~\ref{sec:protocols} specifies two auction protocols with the help of CEDL:  a one-shot combinatorial exchange and a variant of ascending combinatorial auction. We conclude by discussing related work and future work for going further. 

\section{Preliminaries}\label{sec:preliminaries}

To describe a combinatorial exchange, we first define an auction signature, that specifies  the auction participants (the agents), the goods involved in the auction  and  the propositions and variables describing each state of the auction:

\begin{definition}
\label{def:sig}
An \emph{auction signature} $\mathcal{S}$ is a tuple $(N, G, \mathcal{A}, \Phi, Y, I)$, where: 
(i) $N = \{1, 2, \dots, n\}$ is a nonempty finite set of \textit{agents} (or bidders);
(ii) $G = \{1, \dots, m\}$ is a nonempty set of good types;
(iii) $\mathcal{A}$ is a nonempty finite set of \textit{actions} or \textit{bid-trees};
(iv)  $\Phi = \{P, Q, \dots \}$ is a finite set of atomic propositions specifying individual features of a state; 
(v) $Y = \{y_1, y_2, \dots\}$ 
is a finite set of numerical variables specifying  numerical features in a state;
(vi) $I = \{z: z_{min} \leq z \leq z_{max} $ \& $z \in \mathbb{Z}\}$ is a finite subset of integer numbers, denoting the value range for any countable component of the framework, for some arbitrary bounds $z_{min}\leq z_{max}$.  
We denote $I_+ = I\cap \mathbb{N}$ and $I_- = (I\setminus I_+)\cup \{0\}$ as the non-negative and non-positive subsets of $I$, respectively. 
\end{definition}

We will fix an auction signature  $\mathcal{S}$  and all concepts will be based on this signature, except if stated otherwise.
Note that $z_{min}$ and $z_{max}$, in the definition of $I$, should be large enough to represent the total supply of goods being traded, \emph{i.e.}, as we shall see below, $\sum_{i\in N} \sum_{j\in G} x_{i,j}$, as well as the cumulative available money among agents. Through the rest of this paper, we assume that is the case.\footnote{In practice, we could have different bounds for each framework component (e.g., payment, trade).
For notation simplicity, we assume a unique value range $I$. 
}

A \textit{joint allocation} is a tuple $X = (x_{1}, \dots, x_{n})$, where $x_i = (x_{i,1}, \dots$, $x_{i,m})$ is an \textit{individual allocation} for agent $i \in N$ and $x_{i, j} \in I_+$ denotes the number of units $j\in G$ held by $i$. A \textit{joint trade} is a tuple $\Lambda = (\lambda_{1},  \dots, \lambda_{n})$,  where $\lambda_i = (\lambda_{i,1}, \dots, \lambda_{i,m})$ is an \textit{individual trade} for agent $i\in N$ and $\lambda_{i, j} \in I$ denotes the number of units $j\in G$ being traded by agent $i$. 
A trade can be seen as a change over an agent's initial allocation, resulting in a new one. A positive trade expresses how many units of a good type were purchased and a negative trade represents how many units were sold. 

Given an auction signature, we now  define a bidding language allowing agents to express the combination of goods they are willing to buy (or sell) and the value they intend to pay (or receive).

\section{Tree-Based Bidding Language}\label{sec:tbbl}

The Tree-Based Bidding Language (TBBL) \cite{Parkes2005,lubin2008ice} is a language designed for  Combinatorial Exchange. It allows to represent  buyers and sellers demands in the same structure. We adopt TBBL as it is a highly expressive and compact language for combinatorial bids as stressed out by Cavallo \emph{et al.} \cite{cavallo2005tbbl} who compare TBBL to alternative languages such as the OR language \cite{Boutilier01}. 

Our bidding language $\mathcal{L}_{TBBL}$ 
only differs from the original definition of TBBL
in the fact that we assume all language components and related optimization problems are bounded by $I$. 
\begin{definition}\label{def:tbbl}
A formula in $\mathcal{L}_{TBBL}$ is called a \emph{bid-tree} (or simply a \textit{bid}) and is generated by the following BNF:
\[ \beta::= \langle z, j, z\rangle \mid IC_x^y(\bar{\beta}, z)  \] 
where
$\bar{\beta}::= \bar{\beta}, \beta \mid \beta$ 
is a nonempty \textit{bid list},  
$j\in G$, $z \in I$, and $x,y \in I_+$.  
\end{definition}

A bid in the form $\langle q,j,v\rangle$ is called a leaf and represents that the agent is willing to buy (or sell) $q$ units of the good $j$ and pay (or receive) $v$. 
The \textit{interval-choose} (IC) operator defines a range on the number of child nodes that must be satisfied. Thus, a bid $IC_x^y( \bar{\beta}, v)$ indicates the agent is willing to pay (or receive) $v$ for the satisfaction of at least $x$ and at most $y$ of his children nodes $\bar{\beta}$. The IC operator can express  logical connectors. For instance, 
$IC^1_1(\bar{\beta}, v)$ is equivalent to an $XOR$ operator between the bids in the list $\bar{\beta}$. Let $s = |\bar{\beta}|$ (\emph{i.e.}, the list size), $IC^s_s(\bar{\beta}, v)$ is equivalent to an AND operator and $IC^1_s(\bar{\beta}, v)$ is equivalent to an OR operator. For simplicity, we denote $XOR(\bar{\beta}, v)$, (resp. $AND(\bar{\beta}, v)$,  $OR(\bar{\beta}, v)$) as a shortcut for $IC^1_1(\bar{\beta}, v)$ (resp. $IC^s_s(\bar{\beta}, v)$, $IC_1^s(\bar{\beta}, v)$).

For instance, in Figure~\ref{fig:tbbl}, agent $r1$  bids to buy $1$ or $2$ units of $\mathsf{a}$ paying $2$ for each or to sell one unit of $\mathsf{b}$ receiving $3$. Agent  $r2$ bids an exclusive disjunction for either (i) to sell one unit of $\mathsf{a}$ and receive $3$; or (ii) to  sell $2$ units of $\mathsf{a}$ receiving $4$ and to buy one unit of $\mathsf{b}$ paying $2$. The node representing (ii) has an additional value of~$1$. 

\begin{figure}[ht]
\centering
\includegraphics[width=0.8\columnwidth]{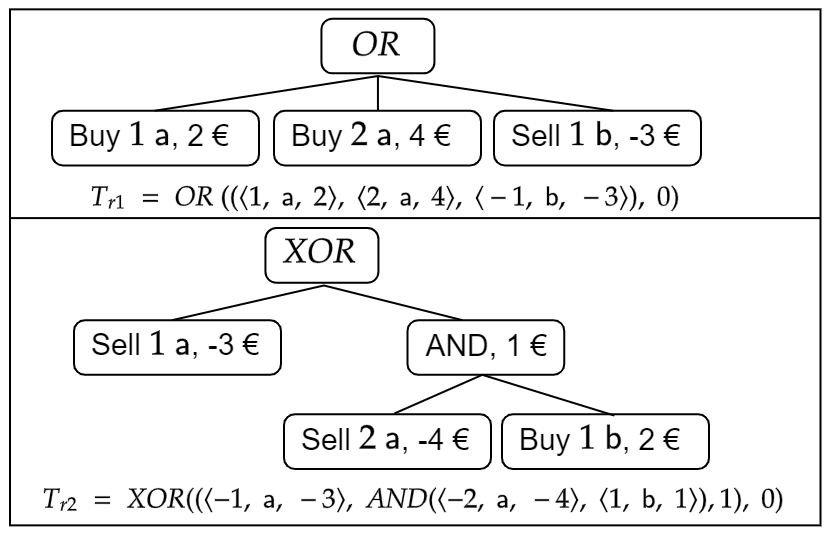}
\caption{Examples of tree-bids for agents $r1$ and $r2$} \label{fig:tbbl}
\Description{}
\end{figure}

Hereafter, we introduce some extra notations to characterize solutions and winners. Let $T_i \in \mathcal{L}_{TBBL}$ be a bid-tree from bid\-der~$i$, 
the set  $Node(T_i)$ denotes all nodes in the tree, that is, all its inner bids, including $T_i$ itself. Formally, if $T_i$ is in the form $\langle q, j, v\rangle$, then $Node(T_i) = \{T_i\}$. Otherwise, $T_i$ is in the form $IC_x^y(\bar{\beta},v')$ and  $Node(T_i) = \{T_i\}\cup  Node(\bar{\beta}_1)\cup \cdots  \cup Node(\bar{\beta}_{s})$, where $s=|\bar{\beta}|$. 

Let $\beta \in Node(T_i)$, the set $Child(\beta) \subset Node(T_i)$ denotes the children of node $\beta$.  If $\beta$ is in the form $IC_x^y(\bar{\beta}, v)$, then 
$Child(\beta) = \{\bar{\beta}_1, \dots, \bar{\beta}_s\}$, where $s=|\bar{\beta}|$. Otherwise, $Child(\beta) = \{\}$. The leaves of a bid-tree $T_i$ are denoted by 
$Leaf(T_i) = \{\langle q, j, v\rangle \in Node(T_i): q \in I, v\in I$ \& $j\in G \}$.
The value specified at node $\beta$ is denoted by $v_i(\beta) \in I$. If $\beta$ is in the form $\langle q, j, v\rangle$, then $v_i(\beta) = v$. Otherwise, $\beta$ is in the form $IC_x^y(\bar{\beta},v')$ and  $v_i(\beta) = v'$. Finally, the quantity of units of the good $j$ specified at a leaf  $\beta = \langle q, j, v\rangle$ is denoted $q_i(\beta, j) = q$. For any other $j' \neq j\in G$, $q_i(\beta, j') = 0$. For any node $\beta \not \in Leaf(T_i)$ and $j\in G$, $q_i(\beta, j) = 0$.

If $\beta$ is not a leaf (\emph{i.e.}, $\beta \in Node(T_i)\setminus Leaf(T_i)$), then it is in the form $IC_x^y(\bar{\beta})$ and we denote by $\mathrm{x}_\beta $ and $\mathrm{y}_\beta$ the interval-choose constraints $x$ and $y$, respectively.

\subsection{Trade value and valid solutions} 
Given a tree $T_i$ from agent $i$, the value of a trade $\lambda_i \in I^m$ is defined as the sum of the values in all satisfied nodes, where the set of satisfied nodes is chosen to provide the \textit{maximal} total value. Let $sat_i(\beta) \in \{0,1\}$ denote whether a node $\beta \in Node(T_i)$ is satisfied and $sat_i = \{\beta: sat_i(\beta) = 1, $ for all $\beta\in Node(T_i)\}$ denote the nodes satisfied in a solution.

A solution $sat_i$ is valid for a tree $T_i$ and trade $\lambda_i$, written $sat_i\in valid(T_i, \lambda_i)$ if Rules \ref{rule:R1} and  \ref{rule:R2} hold \cite{lubin2008ice}:
\begin{gather}
\mathrm{x}_\beta sat_i(\beta) \leq \sum_{\beta' \in Child(\beta)} sat_i(\beta') \leq \mathrm{y}_\beta sat_i(\beta)
\nonumber 
\\ \forall \beta  
\in Node(T_i)\setminus Leaf(T_i)  
\tag{R1}\label{rule:R1}
\nonumber
\\ \sum_{\beta \in Leaf(T_i)} q_i(\beta, j) sat_i(\beta) \leq \lambda_{i,j}, \forall j\in G 
\tag{R2}\label{rule:R2} 
\end{gather}
Rule \ref{rule:R1} ensures that no more and no less than the appropriate number of children are satisfied for any node that is satisfied. Rule \ref{rule:R2} requires that the total increase in quantity of each item across all satisfied leaves is no greater than the total number of units traded.

The total value of a trade $\lambda_i$, given a bid-tree $T_i$, is defined as the solution to the following problem:
\begin{gather}
v_i(T_i, \lambda_i) = \mathop{\argmax}_{sat_i} \sum_{\beta \in Node(T_i)} v_i(\beta) sat_i(\beta)
\nonumber
\\ \text{s.t. } \ref{rule:R1}, \ref{rule:R2} \text{ hold}
\nonumber
\end{gather}

\subsection{Winner Determination}

Given an auction signature, the bid-trees $T = (T_1, \dots, T_n)$ and a joint  allocation $X = (x_1, \dots, x_n)$, where $T_i \in \mathcal{L}_{TBBL}$ denotes the bid-tree from agent $i \in N$,  $x_i = (x_{i, 1}, \dots, x_{i, m})$ is the individual allocation for $i$ and $x_{i,j} \in I_+$  for each good $j\in G$. 

The Winner Determination (WD) defines a pair $(\Lambda, sat)$ obtained by the solution to the following mixed-integer program \cite{lubin2008ice}:
\begin{gather}
WD(T, X): \mathop{\argmax}_{\Lambda, sat} \sum_{i \in N} \sum_{\beta \in Node(T_i)}
v_i(\beta) sat_i(\beta)
\nonumber
\\ \text{s.t. } \lambda_{i, j}  + x_{i,j} \geq 0, \forall i \in N, j\in G
\tag{C1} \label{const:C1}
\\
\sum_{i\in N} \lambda_{i,j} = 0, \forall j\in G
\tag{C2} \label{const:C2}
\\
sat_i \in valid(T_i, \lambda_i), \forall i \in N
\tag{C3} \label{const:C3}
\\
sat_i(\beta) \in \{0,1\}, 
\lambda_{i, j} \in I
\tag{C4} \label{const:C4}
\end{gather}
where $sat= (sat_1, \dots, sat_n)$. Constraint \ref{const:C1} ensures that the joint trade $\Lambda$ is \textit{feasible} given $X$, that is no agent sells more items then she initially hold.
Constraint \ref{const:C2} imposes strict balance in the supply and demand of goods. 
Constraint \ref{const:C3} ensures that each individual trade for an agent $i$ is valid given her bid-tree. Constraint \ref{const:C4} defines the range for trades and node satisfaction. We denote by $\Lambda(T, X)$ the joint trade $\Lambda$ in the solution $WD(T,X)= (\Lambda, sat)$, where $\lambda_{i}(T, X)$ denotes the individual trade for agent $i$ and $\lambda_{i,j}(T, X)$ denotes the units of good $j$ traded by agent $i$. 

If there are two or more solutions for $WD(T, X)$, the trade $\Lambda(T, X)$ will be chosen w.r.t. some total order among the elements of $I^{mn}$. 
This tie-breaking order is omitted to avoid overloading the notation.
In the examples, we assume this order is defined such as it is compatible with the Pareto dominance relation \cite{Voorneveld2003}.

\paragraph{Bid-tree Equivalence}

To ensure the finiteness of the framework, we define a maximal subset of TBBL, such that there is no pair of bid-trees representing the exact same value for every  trade.

\begin{definition}
Given the bid-trees $T_i, T_i' \in \mathcal{L}_{TBBL}$ from agent $i \in N$, we say $T_i$ and $T_i'$ are equivalent bid-trees, denoted $T_i \approx_{tree} T_i'$, if 
for all $\lambda_i \in I^m$, 
$v_i(T_i, \lambda_i) = v_i(T_i', \lambda_i)$.
\end{definition}
Notice that $\approx_{tree}$ is reflexive, symmetric and transitive. 

\begin{proposition}
\label{prop:tbblfinite}
Given $\mathcal{L}_{TBBL}' \subset \mathcal{L}_{TBBL}$, if  for all bid-trees $\beta, \beta' \in \mathcal{L}_{TBBL}'$ where $\beta \neq \beta'$,
$\beta \not \approx_{tree} \beta'$ then $\mathcal{L}_{TBBL}'$ is finite. 
\end{proposition}
\begin{proof}
Let $i \in N$ be an agent and $s = |I|$ denote the size of $I$.  
Assume  $\mathcal{L}_{TBBL}' \subset \mathcal{L}_{TBBL}$ such that  
for all $\beta_i, \beta_i' \in \mathcal{L}_{TBBL}$, where $\beta_i\neq \beta_i'$, we have $\beta_i \not \approx_{tree} \beta_i'$. 
That is, $v_i(\beta_i, \lambda_i) \neq v_i(\beta_i', \lambda_i)$ for some $\lambda_i \in I^m$.

For each trade $\lambda_i \in I^m$, there are $s$ possibilities of distinct values $v_i(\beta_i, \lambda_i)$. Furthermore, there are $s^m$ distinct trades.  
Thus, there may be at most  $s^{s^m}$ non equivalent bids-tree in $\mathcal{L}_{TBBL'}$ and  $\mathcal{L}_{TBBL'}$ is finite.
\end{proof}

Let $\mathcal{L}_{TBBLf}$ be a maximal subset of $\mathcal{L}_{TBBL}$ such that $\beta \not \approx_{tree} \beta'$, for all $\beta, \beta' \in \mathcal{L}_{TBBL}$, where $\beta \neq \beta'$. 
\begin{corollary}
$\mathcal{L}_{TBBLf}$ is finite and  $|\mathcal{L}_{TBBLf}| = |I|^{|I|^m}$.
\end{corollary}

A maximal subset of TBBL without equivalent bids is finite, but its size grows exponentially over the size of $I$ and the quantity $m$ of goods in the auction.

\section{Combinatorial Exchange Description Language}\label{sec:cedl}

The \emph{Combinatorial Exchange Description Language} (CEDL) is a framework for specification of auction-based markets 
and it is composed by a State-Transition Model and a logical language. Next, we present the model, the common  legality constraints and the language's syntax and semantics.
 
\subsection{State-Transition Model}
Given an auction signature and the bidding language, we define the auction protocol through a state-transition model. It allows us to represent the key aspects of an auction, 
at first the legal bids and how to update the auction state. 

Let $\mathcal{A} \subset \mathcal{L}_{TBBL}$ be a finite set of actions. It is not a  limitation to assume $\mathcal{A}$ finite since a maximal subset of TBBL without equivalent bid-trees is finite.
We denote $noop =_{def} \langle 0, g, 0 \rangle$ as the action of not bidding for any good, for some arbitrary $g \in M$.  

\begin{definition}
\label{def:model2}
A state-transition ST-\textit{model} $M$ is a tuple $(W, \bar{w}, T, L$, $U,  \pi_\Phi , \pi_{Y})$, where: 
(i)  $W$ is a nonempty set of \textit{states};
(ii) $\bar{w} \in W$ is the \textit{initial} state;
(iii) $T \subseteq W$ is a set of \textit{terminal} states;
(iv) $ L \subseteq W \times N \times \mathcal{A}$  is a \textit{legality} relation, describing the legal actions at each state, let $L(w, i) = \{a \in \mathcal{A} \mid (w,i,a) \in L \}$ be the set of all legal actions for agent $i$ at state $w$;   
(v) $U : W \times \mathcal{A}^n \rightarrow W$ is an \textit{update} function, 
given $d\in \mathcal{A}^n$, let $d(i)$ be the individual action for agent $i$ in the joint action $d$;
(vi) $\pi_\Phi : W \to 2^\Phi $ is the valuation function for the state propositions;
(vii) $\pi_{Y}: W \times Y \to \mathbb{Z}$, is the valuation function  for the numerical variables.
\end{definition}

\begin{definition}
\label{def:path}

Given an ST-model $M = (W, \bar{w}, T, L, U,$ $ \pi_\Phi , \pi_Y)$, a \textit{path} is a sequence of states and joint actions  $\bar{w}\accentset{d_1}{\to} w_1 \accentset{d_2}{\to}\dots \accentset{d_t}{\to} w_t \accentset{d_{t+1}}{\to} \dots$ such that for any $t \geq 1$:   (i) $w_0 = \bar{w}$; (ii) $w_t \neq w_0$; (iii) $d_t(i) \in L(w_{t-1})$ for any $i \in N$, (iv) $w_t = U(w_{t-1}, d_t)$; and (v) if  $w_{t-1} \in  T$, then $w_{t-1} = w_t$.
\end{definition}

Let $\delta[t]$ denote the $t$-th reachable state of $\delta$, $\theta(\delta, t)$ denote the joint action performed at stage $t$ of $\delta$; and $\theta_i(\delta, t)$ denote the action of agent $i$ performed at stage $t$ of $\delta$, and $\delta[0, t]$ denote the finite prefix $\bar{w}\accentset{d_1}{\to} w_1 \accentset{d_2}{\to}\dots \accentset{d_t}{\to} w_t$.
A path $\delta$ is \textit{complete} if $\delta[e] \in T$, for some $e > 0$. After reaching a terminal state $\delta[e]$, for any $e' > e$, $\delta[e'] = \delta[e]$. 

\subsection{Language Syntax}

Each payment, allocation and trade should be represented as a numerical variable $y \in Y$.
We assume the predefined variables set $\{payment_i, alloc_{i,j}, trade_{i,j} : i \in N, j\in G\} \subseteq Y$. Let $z \in \mathcal{L}_z$ be a numerical term defined as follows:
\begin{align*} 
z::= z' \mid add(z, z) \mid sub(z,z) \mid min(z,z) \mid  max(z,z) \mid  times(z,z)\mid  
\\
 y \mid  win_{i,j}(\bar{\beta}, \bar{Z}) \mid value_i(\beta) \mid value_i(\beta, \bar{z}) \mid qtd_i(\beta, j)
\end{align*}
where $z' \in I, y\in Y$, $i\in N$, $j\in G$, $\beta \in \mathcal{A}$, $\bar{\beta} \in \mathcal{A}^n$  $\bar{z}\in \mathcal{L}_z^m$, and $\bar{Z} \in \mathcal{L}_z^{mn}$.  

The terms $add(z_1, z_2)$, $sub(z_1,z_2)$, $times(z_1$, $z_2)$, $min(z_1,z_2)$ and $max(z_1,z_2)$ specify the corresponding  mathematical operation or function.
For agent $i$ and good $j$, the value of a bid $\beta$, the value of $\beta$ given a trade $\bar{z}$, the quantity of $j$ in bid $\beta$ and the trade $\lambda_{i,j}(\bar{\beta}, \bar{Z})$ are denoted $value_i(\beta)$,  $value_i(\beta,\bar{z})$, $qtd_i(\beta, j)$ and $win_{i,j}$, resp.

The Combinatorial Exchange Description language is denoted by $\mathcal{L}_{CEDL}$ and a formula $\varphi$ in  $\mathcal{L}_{CEDL}$ is defined by the following BNF grammar: 
\begin{align*} 
\varphi ::= p \mid initial \mid terminal \mid legal_i(\beta) \mid does_i(\beta)  \mid 
\\  \neg \varphi \mid \varphi \land \varphi  \mid  \bigcirc \varphi  \mid  z < z  \mid z>z \mid z=z  
\mid 
rest_i(res, \beta)
\end{align*}
where $p \in \Phi$, $i\in N$, 
$res \in \{buyer, seller, good, unit\}$, $\beta \in \mathcal{A}$ and  $z \in \mathcal{L}_z$. 

Intuitively, $initial$ and $terminal$ specify the initial  terminal states, resp.; $legal_i(\beta)$ asserts that agent $i$ is allowed to take action $\beta$ at the current state and $does_i(\beta)$ asserts that agent $i$ takes action $\beta$ at the current state. The formula $\bigcirc \varphi$ means ``$\varphi$ holds at the next state''. The formulas $z_1 > z_2$, $z_1<z_2$, $z_1=z_2$ mean that a numerical term $z_1$ is greater, less and equal to a numerical term $z_2$, resp. 
The formula $rest_i(res, \beta)$ specifies whether the bid $\beta$ from agent $i$ respects the restriction $res \in \{buyer, seller, good, unit\}$. The restriction $buyer$ specifies that $\beta$ cannot have negative quantities or prices. Similarly, the restriction $seller$  specifies that $\beta$ cannot have positive quantities or prices. The restriction $good$ states that $\beta$ should be a leaf node. Finally, the restriction $unit$ says any leaf node in $\beta$ can only demand a single unit from a good type.
 
Other connectives $\lor, \to, \leftrightarrow, \top$ and $ \bot $ are defined by $\neg$ and $\land$ in the standard way.
The comparison operators $\leq$, $\geq$ and $\neq$ are defined by $\lor,  >, < $ and $ =$.  
The extension of the comparison operators $>, <, =$, $\leq$, $\geq$, $\neq$ and numerical terms $max(z_1, z_2), min(z_1, z_2)$, $add(z_1, z_2)$, $sub(z_1, z_2)$, $times(z_1, z_2)$ to multiple arguments is straightforward. 

We define $trade_i$ such that it denotes the list of numerical terms  representing $i$'s trades as follows:
$trade_i =_{def} trade_{i,1}, \dots, trade_{i,m}$. 
The list $alloc^i$ is defined in a similar way. Note that  $trade_i$, $ alloc_i$ $ \in \mathcal{L}_z^m$. 
Given a joint allocation $X$ and a list of bids $T$, 
we write $trade_i = win_i(X, T)$ to denote the formula $\bigwedge_{j\in G} trade_{i,j} = win_{i,j}(X, T)$.
Assume the bids $T = (T_1, \dots, T_n)$, where $T_i \in \mathcal{A}$ is a bid associated to the agent $i \in N$. The formula $does(T) =_{def} \bigwedge_{i \in N} does_i(T_i)$ represents that the agents perform the joint action $T$.

\subsection{Language Semantics}

The semantics for the CEDL language is given in two steps. First, we define Function $f$ to compute the meaning of numerical terms $z \in \mathcal{L}_z$ in some specific  state. Next, a formula $\varphi \in \mathcal{L}_{CEDL}$ is interpreted with respect to a step in a path.

\begin{definition}
\label{def:funcionv} 
Given an ST-model $M$, 
define Function $f:  \mathcal{L}_z \times W  \rightarrow \mathbb{Z}$, assigning any $z \in \mathcal{L}_z$ and state $w \in W$ to a number in $\mathbb{Z}$: 

If $z$ is in the form $add(z', z'')$, $sub(z', z'')$, $min(z', z'')$, $max(z', z'')$, or $times(z', z'')$,
then $f(z,w)$ is defined
through the application of the corresponding mathematical operators and functions over $f(z', w)$ and $f(z'', w)$.
Otherwise, $f(z,w)$ is defined as follows: 
\[f(z, w) = \begin{dcases}  z  \text{ if }  z \in \mathbb{Z}
\\ \pi_{Y}(w, z) \text{ if } z \in Y
\\ \lambda_{i,j}(T, X)  \text{ if }
\\ \quad z = win_{i,j}(T, X)
\end{dcases}
f(z, w) = 
\begin{dcases} 
v_i(T_i) \text{ if } z = value_i(T_i) 
\\ v_i(T_i, \bar{z}) \text{ if } z = value_i(T_i, \bar{z})
\\ q_i(T_i, j) \text{ if } z = qtd_i(T_i, j) 
\end{dcases}\]
\end{definition}

\begin{definition}
\label{def:semantics}
Let $M$ be an ST-Model. Given a path $\delta$ of $M$, a stage $t$ on $\delta$ and  
a formula $\varphi \in \mathcal{L}_{CEDL}$, we say $\varphi$ is true (or satisfied) at $t$ of $\delta$ under $M$, denoted by $M, \delta, t \models \varphi$, according to the following definition:
\[ \arraycolsep=1pt
\begin{array}{lllll}
M, \delta, t \models p &  
& \text{iff} & 
 &  p \in \pi_\Phi (\delta [t]) \\

M, \delta, t  \models \neg \varphi  &  &\text{iff}  &&  M, \delta, t \not\models \varphi \\

M, \delta, t \models \varphi_1 \land \varphi_2  &  &\text{iff}  &&  M, \delta, t \models \varphi_1 \text{ and } M, \delta, t \models \varphi_2 \\

M, \delta, t \models initial  &  &\text{iff}  &&  \delta[t] = \bar{w} \\

M, \delta, t  \models terminal  &  &\text{iff}  &&  \delta[t] \in T\\

M, \delta, t  \models legal_i(a)  &  &\text{iff}  &&  a \in L(\delta[t], i) 
\\

M, \delta, t \models  does_i(a)  &  &\text{iff}  &&  \theta_{i} (\delta, t) = a
\\

M, \delta, t \models \bigcirc \varphi  &  &\text{iff}  &&  
M, \delta, t+1 \models \varphi \\

M, \delta, t  \models  z_1 >z_2 &  & \text{iff} &  &  f(z_1, \delta [t]) > f(z_2, \delta [t]) \\

M, \delta, t \models  z_1 < z_2 &  & \text{iff} &  &   f(z_1, \delta [t]) 	< f(z_2, \delta [t])\\

M, \delta, t  \models  z_1 =  z_2 &  & \text{iff} &  &   f(z_1,\delta [t]) = f(z_2, \delta [t])\\

M, \delta, t \models rest_i(buyer, a)&  &\text{iff}  && \forall \beta \in Node(a), v_i(\beta) \geq 0 \text{ and } 
\\ && && 
 \forall  l \in Leaf(a), \exists j\in G, q_i(l, j) \geq 0
\\
M, \delta, t \models rest_i(seller, a)&  &\text{iff}  && \forall \beta \in Node(a), v_i(\beta) \leq 0 \text{ and } 
\\ && && 
 \forall  l \in Leaf(a), \exists j\in G, q_i(l, j) \leq 0
\\
M, \delta, t \models rest_i(good, a)&  &\text{iff}  && Child(a) = \{\}
\\
M, \delta, t \models rest_i(unit, a)&  &\text{iff}  && \forall \beta \in Leaf(a), \exists j\in G, q_i(\beta, j) \in \{-1, 1\} 
\end{array}\]
\end{definition}

A formula $\varphi$ is globally true through $\delta$, denoted by $M, \delta \models \varphi $, if $M, \delta, t \models \varphi$ for any stage $t$ of $\delta$. A formula $\varphi$ is \textit{globally true} in an ST-Model $M$, written $M \models \varphi$, if $M, \delta \models \varphi$ for all paths $\delta$ in $M$. 
Finally, let $\Sigma$ be a set of formulas in $\mathcal{L}_{CEDL}$, then $M$ is a \textit{model} of $\Sigma$ if $M \models \varphi$ for all $\varphi \in \Sigma$.

The following propositions show that if a player bids at a stage
in a path, then (i) she does not  bid anything else at the same stage and (ii) the bid is legal. Additionally, the value of the bid-tree $noop$ is zero. Notice that an agent bidding $noop$ does not imply her payment will be zero (\emph{e.g}, there may be fees for participating).

\begin{proposition}
\label{prop:general}
For each agent $i \in N$ and each bid-tree $\beta \in \mathcal{A}$,
\begin{enumerate}
    \item \label{prop:general1} $M \models does_i(\beta) \rightarrow \neg does_i(\beta')$, for any $\beta' \in \mathcal{A}$ such that $\beta' \neq \beta$
    \item \label{prop:general2} $M \models does_i(\beta) \rightarrow  legal_i(\beta)$
\end{enumerate}    
\end{proposition}
\begin{proof}
For Statement \ref{prop:general1}, assume $M, \delta, t \models does_i(\beta)$ iff $\theta_r(\delta, t) = \beta$. Then for any $\beta' \neq \beta \in \mathcal{A}_{ce}$, $\theta_r(\delta, t) = \beta'$ and $M, \delta, t \models \neg does_i(\beta')$.

Let us verify Statement \ref{prop:general2}. Assume $M, \delta, t \models does_i(\beta)$, then $\theta_r(\delta,t) = \beta$ and by the definition of $\delta$, $\beta \in L(\delta[t], i)$ and $M, \delta, t \models legal_i(\beta)$.
\end{proof}

\begin{lemma}
\label{lemma:noop}
For each agent $i \in N$, each bid-tree $\beta \in \mathcal{A}$ and each $\bar{z} \in I^m$, $M \models value_i(noop) = 0 \land value_i(noop, \bar{z}) = 0$.     
\end{lemma}
\begin{proof}
We consider Statement \ref{prop:general1}. Remind $noop$ denotes a leaf bid $\langle 0,g,0\rangle$, where $g \in M$. Thus, $v_i(noop) = 0$, $f(value_i(noop), \delta[t]) = 0$ and $M, \delta, t \models value_i(noop) = 0$. Let $\bar{z}\in I^m$. The value of $\bar{z}$ given $noop$, i.e., $v_i(noop, \bar{z})$, is the maximal sum of $v_i(\beta)sat_i(\beta)$ in a solution $sat_i$, for all $\beta \in Node(noop)$. Since  $Node(noop) = \{noop\}$ and  $value_i(noop)=0$, for any solution $sat_i$,  $v_i(noop, \bar{z}) = 0$ and $M, \delta, t \models value_i(noop, \bar{z}) = 0$.
\end{proof}

\section{Representing Mechanism Properties and Auction-based Protocols}\label{sec:protocols}

Let us first show how to represent some classical but important properties from Mechanism Design, namely budget-balance,  no-deficit and individual rationality conditions.

\paragraph{Budget-Balanced Mechanisms}
 
A mechanism is \textit{Budget-Balanced} (BB) if the cumulative payment among the bidders is zero, for every valuations they may have \cite{mishra2018simple}. Given an ST-model $M$, this condition is denoted by the validity    of the following formula:
$BB =_{def} add(payment_1, \dots, payment_n) =0$.

\paragraph{No-deficit Mechanisms}
A mechanism where only the designer can earn revenue  satisfies \textit{no-deficit} \cite{mishra2018simple}. The no-deficit condition is a relaxation from BB, where the cumulative payment among the bidders cannot be negative. 
In CEDL, an ST-model $M$ satisfies the no-deficit condition according to the validity of the following formula: 
$noDeficit =_{def} add(payment_1, \dots, payment_n) \geq 0$.

 \paragraph{Individual Rationality}

A mechanism is individually rational if agents can always achieve at least as much utility as from participating as without participating  \cite{parkes2001iterative}. To represent such condition, we assume each agent $i \in N$ has a private valuation in $I$ for each individual trade $\lambda_i \in I^m$, denoted $\vartheta_i(\lambda_i)$.
As \cite{lubin2008ice}, we also assume the agents have monotonic valuation, so that  $\vartheta_i(\lambda_i') \geq \vartheta_i(\lambda_i)$, for any trade $\lambda_i' \geq \lambda_i$ 
(\emph{i.e},  $\vartheta_{i,j}(\lambda_{i,j}') \geq \vartheta_{i,j}(\lambda_{i,j}$), for each $j$). The agent's utility is quasi-linear, denoted $u_i(\lambda_i, p_i) = \vartheta_i(\lambda_i) - p_i$, where $p_i$ denotes $i$'s payment. 
Rephrased in terms of ST-model, we say a model $M$ is \textit{Individual Rational} $IR$ if it is Individual Rational 
$IR_i$ for each agent $i \in N$ in each path  $\delta$ in $M$ and stages $t$ in $\delta$. 
 
A stage $t$ of $\delta$ is $IR_i$, written $M, \delta, t \models IR_i$ if there is a path $\delta'$ in $M$ such that $\delta[0, t] = \delta'[0, t]$,  $\theta_r(\delta, t) = \theta_r(\delta', t)$, for all $r \in N\setminus \{i\}$,  and 
$M, \delta', t \models utility_i(trade_i, payment_i) = x \rightarrow \bigcirc  utility_i(trade_i, payment_i) \geq x$, for each $x\in I$. In other words, IR requires meta-reasoning as choices among paths have to be considered. 

\smallskip
Let us now represent in CEDL two types of auction-based markets: a One-Shot Combinatorial Exchange and a Simultaneous Ascending Auction. For both of them, we detail the rules representation, the semantic representation and we revisit the Mechanism Design conditions.

\subsection{Representing a Combinatorial Exchange}

To represent a One-Shot Combinatorial Exchange with multiple units of $\mathsf{m}$ goods and $\mathsf{n}$ players, 
we first describe the auction signature, written $\mathcal{S}_{ce} = (N_{ce},G_{ce}, \mathcal{A}_{ce}, \Phi_{ce}, Y_{ce}, I_{ce})$, where $N_{ce} = \{1, \dots, \mathsf{n}\}$, $G_{ce} = \{1, \dots, \mathsf{m}\}$, $\mathcal{A}_{ce} \subseteq \mathcal{L}_{TBBLf}$, $\Phi_{ce} = \{bidRound\}$, 
$Y_{ce} = \{alloc_{i,j}, trade_{i,j}$, $payment_i : i\in N_{ce}, j\in G_{ce}\}$, and $I_{ce} \subset \mathbb{Z}$.

Each instance of a One-Shot Combinatorial Exchange is specific and is defined with respect to $\mathcal{A}_{ce}$, $I_{ce}$ and the constant values $\mathsf{n,m} \in I_{ce,+}\setminus \{0\}$ (the size of $N_{ce}$ and $G_{ce}$, resp.), and $\mathsf{x}_{i,j} \in I_{ce, +}$, for each $i\in N_{ce}$ and $j\in G_{ce}$. 
Each constant  $\mathsf{x}_{i,j}$ represents the quantity of units of $j$ initially held by agent $i$. 
The rules of a One-Shot Combinatorial Exchange are  represented by CEDL-formulas as shown in Figure \ref{fig:CE}. 

\begin{figure}[ht]
\centering
\begin{mdframed}
\begin{enumerate} 
    \item  \label{ce:initial}
    $initial \rightarrow bidRound \land \bigwedge_{i\in N_{ce}} 
(    payment_i = 0$ $\land $
    $ \bigwedge_{j\in G_{ce}}trade_{i,j} = 0)$
   
    \item \label{ce:ter1} $terminal \leftrightarrow \neg initial$

    \item \label{ce:legal1} $terminal \rightarrow legal_i(noop)$, for each $i\in N_{ce}$
    
    \item \label{ce:legal2} $ initial \rightarrow legal_i(\beta)$, for each $i\in N_{ce}$, $\beta \in \mathcal{A}_{ce}$ 
    
    \item \label{ce:alloc}
    $does(T_1, \dots, T_\mathsf{n})\land initial \rightarrow \bigcirc ( \bigwedge_{i \in N_{ce}} trade_i = win_i(T_1, \dots, T_\mathsf{n}, \mathsf{x}_{1, 1}, \dots, \mathsf{x}_{\mathsf{n,m}}))$, for each $(T_1$, $\dots$, $T_\mathsf{n}) \in \mathcal{A}^\mathsf{n}_{ce}$
        
    \item \label{ce:payment} $does_i(\beta) \land initial \rightarrow \bigcirc payment_i = value_i(\beta$, $trade_i)$, for each $i \in N_{ce}$, $\beta \in \mathcal{A}_{ce}$

    \item     
    \label{ce:loop}     
    $terminal \land y = x \rightarrow \bigcirc y= x$, for each $y \in Y_{ce}$, $x \in I_{ce}$ 
    
    \item \label{ce:allocvar} $alloc_{i,j} = add(\mathsf{x}_{i,j}, trade_{i,j}) $, for each $i \in N_{ce}$, $j\in G_{ce}$ 

    \item \label{ce:bidRound} $\bigcirc \neg bidRound$ 
\end{enumerate}
\end{mdframed}
\Description{}
\caption{A Combinatorial Exchange represented by $\Sigma_{ce}$}
\label{fig:CE}
\end{figure}

In the initial state, the trade and payment are zero for every agent and good (Rule \ref{ce:initial}). Any state that is not  initial is  terminal (Rule \ref{ce:ter1}). 
The proposition $bidRound$ helps to distinguish the initial state from the terminal state where no trade or payment were assigned to any agent (\emph{e.g},  when all agents bid $noop$). 
Once in a terminal state, players can only do $noop$. Otherwise, they can bid any bid-tree $\beta \in \mathcal{A}_{ce}$ (Rules \ref{ce:legal1} and \ref{ce:legal2}). 
If a list of bid-trees is the joint action performed in the initial sate, then in the next state each agent receives an individual trade, which is assigned by the WD over the initial allocations and the bid-trees (Rule \ref{ce:alloc}). 
After performing a bid in the initial state, the payment for an agent will be the value of her trade given her bid (Rule  \ref{ce:payment}). 
No numerical variable has its value changed after reaching a terminal state 
(Rule \ref{ce:loop}). The allocation for an agent is the quantity of goods she initially held plus her trade (Rule \ref{ce:allocvar}). Finally, the proposition $bidRound$ is always false in the next state (Rule \ref{ce:bidRound}).
 
\paragraph{Representing as a model}
Next, we address the model representation. 
Let $\mathscr{M}_{ce}$ be the set of ST-models $M_{ce}$ defined for any 
$\mathcal{A}_{ce} \subseteq \mathcal{L}_{TBBLf}$, $I_{ce} \subset \mathbb{Z}$, and the 
constants $\mathsf{m}, \mathsf{n}\in I_{ce,+}\setminus \{0\}$ and $\mathsf{x}_{i,j}\in I_{ce,+}$, for each $i \in N_{ce}$ and $j\in G_{ce}$. 

Each $M_{ce}$ is defined as: 
\begin{itemize}
    \item $W_{ce} = \{\langle b, x_{1,1}, \dots, x_\mathsf{n,m}, \lambda_{1,1}, \dots, \lambda_\mathsf{n,m}, p_1, \dots, p_\mathsf{n}\rangle : b \in \{0$, $1\}, x_{i,j} \in I_{ce, +}$ \& $p_i, \lambda_{i,j} \in I_{ce} $ \& $ i \in N_{ce} $ \& $ j\in G_{ce}\}$; 
    
    \item  $\bar{w}_{ce} = \langle 1, \mathsf{x}_{1,1},  \dots,  \mathsf{x}_\mathsf{n,m}, 0, \dots, 0, 0,\dots, 0\rangle$;

    \item $T_{ce}= \{w:  w\in W_{ce} $ \& $w\neq\bar{w}_{ce}\}$;
    
    \item $L_{ce}=  \{(w, i, noop): i \in N_{ce} $ \& $ w \in T_{ce}\}  \cup  \{ ( \bar{w}_{ce}, i,$ $ \beta) : \beta \in \mathcal{A}_{ce} $ \& $i\in N_{ce}\}$;

    \item $U_{ce}$ is defined as follows: for all $w = \langle b, x_{1,1}, \dots, x_\mathsf{n,m},  \lambda_{1,1}, \dots, $ $\lambda_\mathsf{n,m}, p_1, \dots, p_\mathsf{n}\rangle \in W_{ce}$ and for all 
    $d  \in \mathcal{A}^\mathsf{n}_{ce}$: 
    \begin{itemize}
        \item If $w = \bar{w}_{ce}$, then  
        $U_{ce}(w, d) =  \langle 0, x_{1,1}', \dots, x_\mathsf{n,m}', \lambda_{1,1}', \dots$, $\lambda_\mathsf{n,m}', p_1',  \dots,  p_\mathsf{n}'\rangle$, where 
        $\lambda_{i,j}' = \lambda_{i,j}(d, \mathsf{x}_{1,1}, \dots, \mathsf{x}_\mathsf{n,m})$; 
        $x_{i,j}' = \mathsf{x}_{i,j} +  \lambda_{i,j}'$; and 
        $p_i' = v_i(d(i), \lambda_{i,1}' , \dots, \lambda_{i,\mathsf{m}}')$, for each $i\in N_{ce}$ and $j\in G_{ce}$;  
        
        \item Otherwise, $U_{ce}(w,d) = w$.
    \end{itemize}

    \item For each $w \in W_{ce}$, $ i \in N_{ce}$ and $j\in G_{ce}$, 
    $\pi_{Y, ce}(w$, $trade_{i,j}) = \lambda_{i,j}$;  $\pi_{Y, ce}(w, alloc_{i,j}) = x_{i,j}$; $\pi_{Y, ce}(w, payment_i) = p_i$; and $\pi_{\Phi, ce}(w) = \{bidRound : b = 1\}$.

\end{itemize}

Hereafter, we assume an instance of $M_{ce} \in \mathscr{M}_{ce}$ and $\Sigma_{ce}$ for some $\mathcal{A}_{ce} \subseteq \mathcal{L}_{TBBLf}$, $I_{ce} \subset \mathbb{Z}$, $\mathsf{m}, \mathsf{n} \in I_{ce, +}\setminus \{0\}$ and $ \mathsf{x}_{i,j} \in I_{ce, +}$, where $i\in N_{ce}$, $j\in G_{ce}$.

\begin{example}
\label{example:ce}
Let $M_{ce} \in \mathscr{M}_{ce}$, where (i) $\mathsf{n} = 2$ and the agents are denoted by $r1$ and $r2$, (ii) $\mathsf{m} = 2$ and the good types are denoted by $\mathsf{a}$ and $\mathsf{b}$,  and (iii)  $\mathsf{x}_{r1,\mathsf{a}} =0$, $\mathsf{x}_{r1,\mathsf{b}} =1$, $\mathsf{x}_{r2,\mathsf{a}} =2$ and $\mathsf{x}_{r2,\mathsf{b}} =0$, \emph{i.e}, at the beginning of the auction, agent $r1$ has $1$ unit of $\mathsf{b}$ and agent $r2$ has $2$ units of $\mathsf{a}$. 
Figure~\ref{fig:ceflow} illustrates a path in $M_{ce}$, where the agents perform the bids previously introduced in Figure~\ref{fig:tbbl}. 
In state $w_0$, all the payments and trades are zero. Their joint bid leads to state $w_1$, where the joint trade obtained by the winner determination is $(2, -1,-2,1)$. The tie-breaking ensures that the joint trade is unique. Given the joint trade, the allocation for agent $r1$ is $2$ units of $\mathsf{a}$ and the allocation for agent $r2$ is $1$ unit of $\mathsf{b}$. Since $w_1$ is terminal, the agents can only bid $noop$ and the state that succeeds $w_1$ is $w_1$ itself. 

\begin{figure}[ht]
\centering
\includegraphics[width=1\columnwidth]{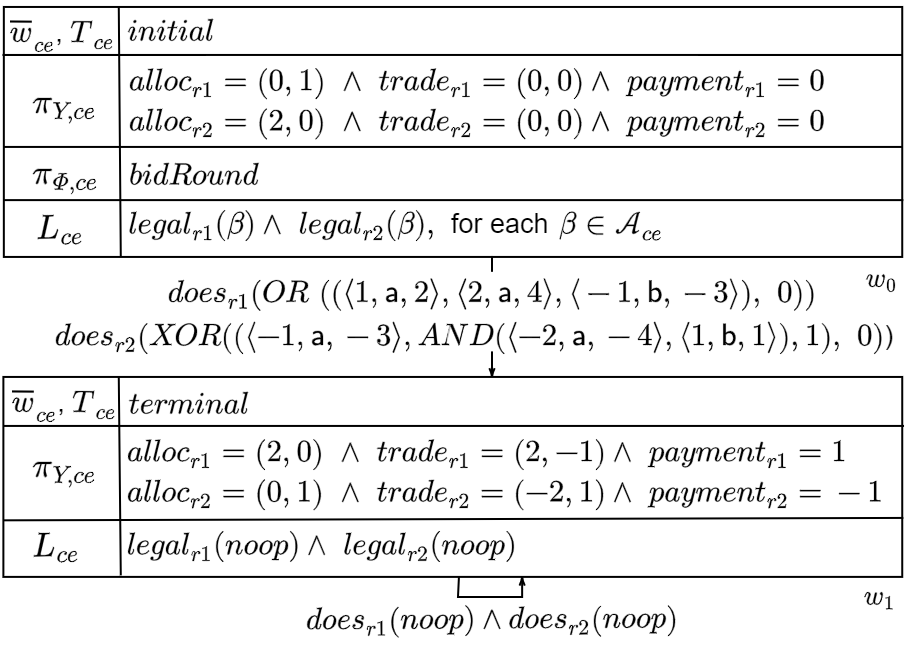}
\caption{A Path in $M_{ce}$, with 2 bidders and $2$ goods} \label{fig:ceflow}
\Description{}
\end{figure}

\end{example}

Let us now evaluate the protocol. First, 
Proposition  \ref{prop:modelce} shows that $M_{ce}$ is a sound representation of $\Sigma_{ce}$ . 

\begin{proposition}
\label{prop:modelce}
$M_{ce}$ is an ST-model and it is a model of  $\Sigma_{ce}$.
\end{proposition}
\begin{proof}\emph{(Sketch)} 
It is routine to check that $M_{ce}$ is actually an ST-model.  Given a path $\delta$ in $M_{ce}$ and a stage  $t$ of $\delta$, we need to show that $M_{ce},\delta, t \models \varphi$, for each $\varphi \in \Sigma_{ce}$.  

Let us verify Rule \ref{ce:initial}. Assume $M_{ce}, \delta, t \models initial$, then $\delta[t] = \bar{w}_{ce}$, i.e., $\delta[t] = \langle 1, \mathsf{x}_{1,1}, \dots, \mathsf{x}_{\mathsf{n,m}}, 0, \dots, 0, 0, \dots, 0 \rangle$. By the definitions of $\pi_{Y, ce}$ and $\pi_{\Phi, ce}$, 
$\pi_{\Phi, ce}(\bar{w}_{ce}) = \{bidRound\}$,  $\pi_{Y, ce}(\bar{w}_{ce}, payment_i) = 0$ and $\pi_{Y, ce}(\bar{w}_{ce}, trade_{i,j}) = 0$ for all $i \in N_{ce}$ and $j\in G_{ce}$. Thus, $M_{ce}, \delta, t \models bidRound \land \bigwedge_{i \in N_{ce}} payment_i = 0 \land \bigwedge_{j\in G_{ce}} trade_{i,j} = 0$.

Now we verify Rule \ref{ce:legal2}. Assume $M_{ce}, \delta, t \models initial$, then $\delta[t] = \bar{w}_{ce}$ and for all $i \in N_{ce}$ and $\beta \in \mathcal{A}_{ce}$, $(\bar{w}_{ce}, i, \beta) \in L_{ce}$. Thus, $M_{ce}, \delta, t \models legal_i(\beta)$.

Then we consider Rule \ref{ce:alloc}. $M_{ce}, \delta, t \models does(T_1, \dots, T_\mathsf{n}) \land initial$, for  $(T_1, \dots, T_\mathsf{n}) \in \mathcal{A}^\mathsf{n}_{ce}$, i.e., $M_{ce}, \delta, t \models \bigwedge_{i \in N_{ce}} does(T_i)$ and $M_{ce}, \delta, t \models initial$. Thus, $\theta_i(\delta, t) = T_i$ and $\delta[t] = w_{ce}$, for all $i \in N_{ce}$. The update function $U_{ce}$ defines $\delta[t+1]$ such that $\pi_{Y,ce}(\delta[t+1], trade_{i,j}) = \lambda_{i,j}(T_1, \dots, T_\mathsf{n}, \mathsf{x}_{1,1}, \dots, \mathsf{x}_{\mathsf{n,m}})$, for each $i \in N_{ce}$ and $j\in G_{ce}$. Thus, $M_{ce}, \delta, t+1 \models \bigwedge_{i \in N_{ce}, j\in G_{ce}} trade_{i,j} = win_{i,j}(T_1, \dots$, $T_\mathsf{n}, \mathsf{x}_{1,1}, \dots$, $\mathsf{x}_{\mathsf{n,m}})$  and also $M_{ce}, \delta, t \models \bigcirc (\bigwedge_{i \in N_{ce}, j\in G_{ce}} trade_{i,j} = win_{i,j}(T_1, \dots$, $T_\mathsf{n}$, $\mathsf{x}_{1,1}, \dots, \mathsf{x}_{\mathsf{n,m}}))$. Using the abbreviation for individual trades, $M_{ce}, \delta$, $t \models \bigcirc(\bigwedge_{i \in N_{ce}} trade_i = win_i(T_1, \dots, T_\mathsf{n}, \mathsf{x}_{1,1}, \dots$, $\mathsf{x}_{\mathsf{n,m}}))$. 

Finally, we consider Rule \ref{ce:allocvar}. 
Let $i \in N_{ce}$ and $j\in G_{ce}$. 
If $t = 0$, then $\delta[t] = \bar{w}_{ce}$. By the valuation function $\pi_{Y, ce}$, $\pi_{Y, ce}(\bar{w}_{ce}, alloc_{i,j}) = \mathsf{x}_{i,j}$ and $\pi_{Y, ce}(\bar{w}_{ce}, trade_{i,j}) = 0$. 
If $t = 1$, then by the path definition we have $\delta[t] = U_{ce}(\bar{w_{ce}}, d)$, for some $d \in \mathcal{A}_{ce}^\mathsf{n}$. The update function  $U_{ce}$ defines $\pi_{Y, ce}(\bar{w}_{ce}, alloc_{i,j}) = \mathsf{x}_{i,j} + trade_{i,j}$. Otherwise, for any $t>1$, $\delta[t] \in T_{ce}$ (see Rule \ref{ce:loop} and path definition). Thereby, $U_{ce}(\delta[t], d) = \delta[t]$, for some joint action $d$.
Thus, $M_{ce}, \delta, t \models alloc_{i,j} = add(\mathsf{x}_{i,j}, trade_{i,j})$.
 
The remaining rules are verified in a similar way.
\end{proof}

Next, we focus on general protocol properties: termination and playability.  
Each state succeeding the initial is terminal. 
Thus, $M_{ce}$ satisfies the termination condition from General Game Playing (GGP) \cite{ZhangAI2018}, that is, each path in  $M_{ce}$ reaches a terminal state. 
It follows that agents can only bid once (except by $noop$).  

\begin{proposition}
For each agent $i \in N_{ce}$ and bid-tree $\beta'$,
\label{prop:cefini}
\begin{enumerate}
    \item \label{prop:cefini1}$M_{ce} \models initial \rightarrow \bigcirc terminal$
    \item \label{prop:cefini2} $M_{ce} \models legal_i(\beta) \rightarrow \bigcirc \neg  legal_i(\beta')$, for any $\beta' \in \mathcal{A}_{ce}$ such that $\beta' \neq noop$
\end{enumerate}
\end{proposition}
\begin{proof}
Given a path $\delta$ in  $M_{ce}$ and a stage $t$ of $\delta$. 
Let us verify Statement \ref{prop:cefini1}. Assume $M_{ce}, \delta, t \models initial$. Then, $\delta[t] = \bar{w}_{ce}$. By the path definition, for any $j\geq 1$, $\delta[j] \neq \bar{w}_{ce}$. By the construction of $T_{ce}$, we have $T_{ce} = W_{ce}\setminus \{\bar{w}_{ce}\}$. Thus, $M_{ce}, \delta, t+1 \models terminal$ and 
$M_{ce}, \delta, t \models \bigcirc terminal$.

Now we verify Statement \ref{prop:cefini2}. Assume $M_{ce}, \delta, t \models legal_i(\beta)$. From Statement \ref{prop:cefini1} and since the path construction defines a loop in any terminal state (i.e, if $\delta[t] \in T_{ce}$ then $\delta[t] = \delta[t+1]$), we have that $\delta[t+1] \in T_{ce}$. Thus, $L_{ce}(\delta[t+1], i) = \{noop\}$ and $M_{ce}, \delta, t+1 \models \neg legal(\beta')$, for any $\beta' \in \mathcal{A}_{ce}$ such that $\beta' \neq noop$. 
\end{proof}

$M_{ce}$ satisfies the playability condition from GGP \cite{ZhangAI2018}, \emph{i.e}, there is always at least one legal action for each agent in any non-terminal stage of a path. 

\begin{proposition}
\label{prop:playability}
For each agent $i \in N_{ce}$,
$M_{ce} \models \bigvee_{a \in \mathcal{A}_{ce}}legal_i(a)$.
\end{proposition}
\begin{proof}
Straightforward from Rules \ref{ce:legal1} and \ref{ce:legal2} from $\Sigma_{ce}$.
\end{proof}

The next lemma shows that if an agent bids $noop$ in an initial state, her payment will be zero. Furthermore, if the payment is zero in a terminal state, it will be  zero in the succeeding state.

\begin{lemma}
\label{lemma:paymentnoop}
For each agent $i \in N_{ce}$, 
\begin{enumerate}
    \item $M_{ce} \models initial \land does_i(noop) \rightarrow \bigcirc payment_i = 0$
    \item $M_{ce} \models terminal \land payment_i =0 \rightarrow \bigcirc payment_i=0$ 
\end{enumerate}
\end{lemma}
\begin{proof}
Straightforward from Corollary \ref{lemma:noop} and Rules \ref{ce:payment} and \ref{ce:loop} from $\Sigma_{ce}$.
\end{proof}

We then focus on budget balance, non-deficit and individual rationality conditions.

\begin{thm}
  $M_{ce} \not \models BB$ and
  $M_{ce}  \models noDeficit$. 
\end{thm}
\begin{proof}
Let $\delta$ be a path in $M_{ce}$ and $t$ be a stage in $\delta$.

If $\delta[t] = \bar{w}_{ce}$, then $M_{ce}, \delta, t \models \bigwedge_{i \in N_{ce}} payment_i = 0$. Thus, $M_{ce}, \delta, t \models BB \land noDeficit$.

Otherwise, by the path definition, $\delta[t] = U_{ce}(\delta[t-1], \theta(\delta, t-1))$. Since $M_{ce} \models initial \rightarrow \bigcirc terminal$ and given that $\delta[t-1] = \delta[t]$ whenever $\delta[t-1]\in T_{ce}$, we focus on the case where $\delta[t-1] = \bar{w}_{ce}$ and the remaining cases follow by consequence.

Let us denote $T_i = \theta_i(\delta, t-1)$, for each $i \in N_{ce}$. By Rules \ref{ce:alloc} and \ref{ce:payment} from $\Sigma_{ce}$, we have $M_{ce}, \delta, t \models \bigwedge_{i \in N_{ce}} (trade_i = win_i(T_i, \dots, T_\mathsf{n}$, $\mathsf{x}_{1,1}, \dots$, $\mathsf{x}_{\mathsf{n,m}}) \land payment_i = value(T_i, trade_i))$. 
The  solution $WD(T_i$, $\dots, T_\mathsf{n}, \mathsf{x}_{1,1}$, $\dots, \mathsf{x}_{\mathsf{n,m}}) = (\Lambda, sat)$ satisfies Constraints \ref{const:C1}-\ref{const:C4} and 
maximizes $\sum_{i \in N_{ce}} \sum_{\beta \in Node(T_i)}  v_i(\beta) sat_i(\beta)$, that is, it maximizes $\sum_{i \in N_{ce}}  v_i(T_i, \lambda_i)$. 
We define the pair $(\Lambda', sat')$, such that $\Lambda' = (0, \dots, 0)$ is an empty joint trade and 
$sat' = (sat_1, \dots, sat_\mathsf{n})$, where $sat_i'=\{\}$,  for all $i \in N_{ce}$. Notice $sat_i'$ is valid for $T_i$ and $\lambda_i'$ (i.e, $sat_i' \in valid(T_i, \lambda_i')$) and $\sum_{\beta \in Node(T_i)}v_i(\beta)sat_i'(\beta) = 0$.
Remind that the value of a trade $\lambda_i'$ given a bid-tree $T_i$ maximizes the value of the satisfied nodes in a valid solution. 
Thus, $v_i(T_i, \lambda_i') \geq 0$ (i.e., it is at least equal to  $\sum_{\beta \in Node(T_i)}v_i(\beta)sat_i'(\beta)$).

Since the pair $(\Lambda', sat')$ satisfies the Constraints \ref{const:C1}-\ref{const:C4}, 
a solution $(\Lambda, sat)$ for WD should have at least the same cumulative trade value, that is, $\sum_{i \in N_{ce} v_i(T_i, \lambda_i)} \geq \sum_{i \in N_{ce} v_i(T_i, \lambda_i')}$. Thus,  $\sum_{i \in N_{ce}} v_i(T_i$, $\lambda_i) \geq 0$. Since $M_{ce}, \delta, t \models \bigwedge_{i \in N_{ce}} payment_i = value(T_i, trade_i)$, we have $M_{ce}, \delta, t \models add(payment_1, \dots, payment_\mathsf{n}) \geq 0$. Thereby, 
$M_{ce}, \delta, t \models noDeficit$.

However, if $\sum_{i \in N_{ce}} v_i(T_i, \lambda_i) > 0$, then $M_{ce}, \delta, t \models add(payment_1$, $\dots, payment_\mathsf{n}) > 0$ and $M_{ce}, \delta, t \not \models BB$.
\end{proof}

\begin{thm}
\label{thm:irce}
For each $i \in N_{ce}$ and some monotonic valuation $\vartheta_i$ over individual trades, 
 $M_{ce} \models IR_i$. 
\end{thm}
\begin{proof}
Given a path $\delta$ in $M_{ce}$, a stage $t$ in $\delta$, an agent $i\in N_{sa}$, and a monotonic valuation $\vartheta_i(\lambda_i) \in I_{ce}$ for each individual trade $\lambda_i \in I^\mathsf{m}_{ce}$, let us consider the case where $\delta[t] = \bar{w}_{sa}$. We have $M_{sa}, \delta, t \models payment_i = 0 \land \bigwedge trade_{i,j} = 0$. Then the utility of $i$ is simply  $\vartheta_i(0, \dots, 0)$. Let $\delta'$ be a path such that $\delta[0,t] = \delta'[0,t]$, $\theta_i(\delta, t) = noop$, and  $\theta_r(\delta, t) = \theta_r(\delta', t)$, for all $r \in N_{ce}\setminus \{i\}$. Since $noop \in L_{ce}(\delta[j],i)$, there is  such path in $M_{ce}$. 
 From Lem\-ma~\ref{lemma:paymentnoop}, we have $M_{ce}, \delta', t+1 \models payment_i = 0$.  Constraint \ref{const:C3} (from the winner determination) and Rule \ref{rule:R2} (from the valid solution definition) ensure that the quantity of each item across the leaves of $noop$ is no greater than the total number of units traded. For each $j\in G_{sa}$, since $q_i(noop, j)=0$ , then $M_{ce}, \delta', t+1 \models trade_{i,j} \geq 0$. Since $i$'s valuation is monotonic, it follows that $M_{ce}, \delta', t+1 \models utility_i(trade_i, payment_i)\geq \vartheta_i(0, \dots, 0)$.
 
Now, let us assume $ \delta[t] \neq \bar{w}_{sa}$. Let $x \in I_{sa}$ such that  $M_{ce}, \delta, t \models utility_i(trade_i, payment_i) = x$. From the path definition,  $\delta[t] = \delta[t+1]$ and thus $M_{ce}, \delta, t+1 \models utility_i(trade_i, payment_i) = x$. 
\end{proof}

\subsubsection{Vickrey–Clarke–Groves payment}
A Vickrey-Clarke-Groves (VCG) mechanism computes a discount for each winner's  payment, such as she has an incentive to be truthful: the bidder is willing to reveal her private value. Remind that this work focuses on the auction definition and not on the bidder's behavior, the reader may refer to Krishna (\citeyear{Krishna}) for this strategic aspect.

Let us  show how to express VCG payments for an agent $i \in N_{ce}$. 
Given the bid-trees $T = (T_1, \dots, T_\mathsf{x})$ and the initial joint allocation $\mathsf{X} \in I^{\mathsf{nm}}_{ce}$, 
let $T' = (T_1', \dots, T_n')$ be defined
as follows: $T_i'=noop$ and $T_r' = T_r$,  for all $r \in N_{ce}\setminus \{i\}$.
Similarly, let $\mathsf{X}' = (x_{1,1}', \dots, \mathsf{x}_{\mathsf{n,m}}')$ be defined as follows: $\mathsf{x}_{i,j}' = 0$ and $\mathsf{x}_{r,j}' = x_{r,j}'$, for all $r \in N_{ce}\setminus \{i\}$ and all $j\in G_{ce}$. 
Remind $\lambda_{i}(T, \mathsf{X})$ denotes the individual trade of agent $i$ in the solution for $WD(T, \mathsf{X})$.

The VCG payment for agent $i$ is the value of the bid-tree $T_i$ given the individual trade $\lambda_{i}(T, \mathsf{X})$ discounted by the difference between the cumulative values from the joint trade $\Lambda(T, \mathsf{X})$ and the trade resulting from removing the bid and allocation of $i$. Formally,
\[p_{vcg,i} = v_i(T_i, \lambda_{i}(T, \mathsf{X})) - \sum_{r \in N_{ce}}( v_r(T_r, \lambda_{r}(T, \mathsf{X})) - v_r(T_r', \lambda_{r}(T', \mathsf{X}')))\]

To construct a combinatorial exchange with VCG payments, we can define $\Sigma_{vcg}$ such that it is defined  exactly as $\Sigma_{ce}$, except by Rule \ref{ce:payment}, which is replaced by the following:
\begin{multline*}
does(T) \land \neg terminal
\land 
p_i = sub(
        value_i(T_i, win_i(T, \mathsf{X})), 
        \\
        add(
            sub(
                value_1(T_1, win_1(T, \mathsf{X})),
                value_1(T_1', win_1(T', \mathsf{X}'))      
            ), \dots, 
            \\
            sub(
                value_\mathsf{n}(T_\mathsf{n}, win_\mathsf{n}(T, \mathsf{X})),
                value_\mathsf{n}(T_\mathsf{n}', win_\mathsf{n}(T', \mathsf{X}'))      
            )
            )    
        )
\\
\rightarrow\bigcirc  payment_i = p_i      
\end{multline*}
for each $p_i \in I_{ce}$, $T \in \mathcal{A}^\mathsf{n}_{ce}$ and  $i\in N_{ce}$. 

Let $M_{vcg}$ be an ST-model defined as $M_{ce}$, except by the definition of 
$\pi_{Y, vcg}(w, payment_i)$, for all $i$, which are defined by according to the VCG price $p_{vcg,i}$ in a state $w \in W_{vcg}$. Unsurprisingly,  $BB$ and $noDeficit$ are not valid in $M_{vcg}$.

\begin{proposition}
\label{prop:vcgnotbb}
$M_{vcg} \not \models BB$ and $M_{vcg} \not \models noDeficit$. 
\end{proposition}
\begin{proof} 
Let us prove it by showing a  counterexample. 
Assume $\mathsf{n} = 2$, $\mathsf{m} = 1$, $\mathsf{x}_{1,1} =0$ and $\mathsf{x}_{2,1} =1$ i.e., there are two agents, one good type and the second agent initially holds one copy of the good. Let $T_1 =  \langle 1,1,2\rangle$ and $T_2 = \langle -1, 1, -1\rangle$
Given a path $\delta$ in $M_{vcg}$, let  $t$ in $\delta$ be a stage such that $M_{vcg}, \delta, t \models initial \land does_1(T_1) \land does_2(T_2)$. That is, in the initial state, agent $1$ bids for buying the good at the price $2$ and agent $2$ bids for selling the good at the price $-1$. 
By WD, we have $\Lambda(T_1, T_2, 0,1) = (1,-1)$. Thus, $v_1(T_1, 1) = 2$ and $v_2(T_2, -1) = -1$. By the update function, we have $M_{vcg}, \delta, t+1 \models trade_{1,1} = 1 \land trade_{2,1} = -1$. 
Note no trade is performed when any of the agents does not participate, i.e., $\Lambda(noop, T_2, 0,1) = \Lambda(T_1, noop, 0,0) = (0,0)$. We also have $v_i(T_i, 0) = v_i(noop, 0) =0$, for each $i \in N_{vcg}$. By the VCG payment rule, $M_{vcg}, \delta, t+1 \models payment_1 = sub(2, add(add(2, -1),add(0,0))) \land payment_2 = sub(-1, add(add(2, -1),add(0,0)))$. That is,  $M_{vcg}, \delta, t+1 \models payment_1 = 1\land payment_2 = -2$. Thus, we have a budget deficit of $-1$, i.e., $M_{vcg}, \delta, t+1 \models add(payment_1, payment_2) = -1$. Thereby, $M_{vcg} \not \models BB$ and $M_{vcg} \not \models noDeficit$.
\end{proof}

\subsection{Representing a Simultaneous Ascending Auction}

Let us now consider a new type of auction: the Simultaneous Ascending Auction (SAA) is a single-side and single-unit auction similar to the traditional English Auction, except that several goods are sold at the same time, and that the bidders simultaneously bid for any number of goods they want 
\cite{Cramton06}. 
To represent a SAA 
with $\mathsf{m}$ good types and $\mathsf{n}$ agents, we first describe the auction signature, written $\mathcal{S}_{sa} = (N_{sa},  G_{sa}, \mathcal{A}_{sa}, \Phi_{sa}, Y_{sa}, I_{sa})$, where $N_{ce} = \{1, \dots, \mathsf{n}\}$, $G_{sa} = \{1, \dots, \mathsf{m}\}$, $\mathcal{A}_{sa} \subseteq \mathcal{L}_{TBBLf}$, $\Phi_{sa} = \{sold_j,  bid_{i,j} :  $ $ j\in G_{sa} $ \& $ i\in N_{sa}\}$ and $Y_{sa} = \{price, price_j, $ $ trade_{i,j}, alloc_{i,j}, payment_{i}: j\in G_{sa} $ \& $ i\in N_{sa}  \}$. The propositions $sold_j$ and $bid_{i,j}$ represent whether the good $j$ was sold and whether $i$ is bidding for $j$, resp. The variables $price$ and $price_j$ specify the current price for any unsold good and the selling price for $j$, resp. 

Each instance of a SAA 
is specific and defined with respect to $\mathcal{A}_{sa}$, $I_{sa}$ and the constant values $\mathsf{inc}, \mathsf{n}, \mathsf{m} \in I_{sa, +}\setminus \{0\}$ and $ \mathsf{start} \in I_{sa, +}$, representing the quantity of agents and goods, the increment, and the starting price, respectively. Let $\mathsf{max}_{sa}$ denote the largest value in $I_{sa}$.
Then, the rules of an SAA 
are formulated by CEDL-formulas as shown in Figure \ref{fig:SA}.

\begin{figure}[ht]
\centering
\begin{mdframed}
Given $p_1, \dots, p_\mathsf{m} \in  I_{sa, +}$, let $or_{p_1, \dots, p_\mathsf{m}} =_{def} OR((\langle 1, 1, p_1 \rangle$,  $\dots, \langle 1, 1, p_\mathsf{m} \rangle), 0)$,
  
\begin{enumerate}

\item \label{sa:initial} $initial \leftrightarrow price = \mathsf{start} $ $ \land  \bigwedge_{j\in G_{sa}} \big( price_j = \mathsf{start}$ $\land$ 
  $ \bigwedge_{i \in N_{sa}}
  (\neg bid_{i,j} \land trade_{i,j} = 0)\big) $
    
\item \label{sa:sold}
$sold_j \leftrightarrow \bigvee_{i \in N_{sa}} trade_{i,j} = 1$, for each $j \in G_{sa}$ 

\item \label{sa:terminal} $terminal \leftrightarrow $
$\neg initial \land \bigwedge_{j\in G_{sa}}  (sold_j $ $ \lor $ $ \bigwedge_{i\in N_{sa}}\neg bid_{i,j} )$ 
 
\item \label{sa:trade1}
  $\bigcirc (trade_{i,j}=1 \leftrightarrow  
  bid_{i,j} $ $\land $    $\bigwedge_{r \in N_{sa} \setminus \{i\}} \neg bid_{r,j} )$, for each $i \in N_{sa}$, $ j \in G_{sa}$
   
\item \label{sa:trade0} $\bigcirc( trade_{i,j}=0 \leftrightarrow \neg ( bid_{i,j} $ $\land $   $\bigwedge_{r \in     N_{sa} \setminus \{i\}} \neg bid_{r,j}))$, for each $i \in N_{sa}$, $j \in G_{sa}$ 
      
\item \label{sa:legal} 
$ legal_i(or_{p_1, \dots, p_\mathsf{m}}) \leftrightarrow$ 
  $\bigwedge_{j\in G_{sa}}\big( (p_j = 0 $ $ \land $ $ trade_{i,j} = 0  )$ $\lor $ $(p_j = price \land  \neg sold_j) $  $\lor $ $(p_j = price_j  \land trade_{i,j} = 1 ) \big )$,
  for each $i \in N_{sa}$, 
  $ p_1, \dots, p_\mathsf{m} \in \{x: 0\leq x <\mathsf{max}_{sa}-\mathsf{inc}\}$
  
\item \label{sa:price}
    $\neg terminal \land price = x \rightarrow \bigcirc  price = add(x, \mathsf{inc})$,   for each $x \in I_{sa,+}$
 
\item \label{sa:pricej}
    $\neg terminal \land price_j = x\rightarrow \bigcirc ((price_j = x \land sold_j) \lor (price_j = add(x, \mathsf{inc}) \land \neg sold_j))$, for each $j\in G_{sa}$,  $x\in I_{sa, +}$ 
\item \label{sa:alloc0}
  $\neg terminal \rightarrow alloc_{i,j} = 0$, for each $i \in N_{sa}$, $j \in G_{sa}$ 
  
\item \label{sa:bid}
  $\bigcirc  bid_{i,j} \leftrightarrow $  $(  does_i(or_{p_1, \dots, p_\mathsf{m}}) $ $\land $ $p_j \neq 0 )   $   $\lor $ $(bid_{i,j} \land terminal)$, for each $i\in N_{sa}$, $j \in G_{sa}$ and some $p_1, \dots, p_m \in I_{sa, +}$ 
  
\item \label{sa:alloctrade}  
  $terminal \rightarrow alloc_{i,j} = trade_{i,j}$, for each $i \in N_{sa}$, $j \in G_{sa}$ 
  
\item \label{sa:payment}
  $payment_i = add(times(price_1, trade_{i,1})$, $\dots$,  $times(price_\mathsf{m}, trade_{i,\mathsf{m}})
  )$, for each $i \in N_{sa}$ 

  \item \label{sa:loop} $terminal \land y= x \rightarrow \bigcirc y = x$, for each $y \in Y$, $x \in I_{sa}$ 
\end{enumerate}
\end{mdframed}
\caption{Simult. Ascending Auction represented by $\Sigma_{sa}$}
\Description{}
\label{fig:SA}
\end{figure}

In the initial state, no agent is bidding, no trade is performed and the prices have the value $\mathsf{start}$ (Rule \ref{sa:initial}). A good is sold if there is a trade for some agent (Rule \ref{sa:sold}). In a terminal state, all the goods are either sold or no one is bidding for them (Rule \ref{sa:terminal}). A good will be traded to an agent if she is the only one bidding for it, otherwise there is no trade (Rules \ref{sa:trade1}-\ref{sa:trade0}). 
For each good, an agent can either bid the value 0, the current price (for unsold goods) or repeat her winning bid for this good (Rule \ref{sa:legal}). In a non-terminal state, the propositions and numerical variables are updated as follows: (i) the current price increases, (ii) 
the selling price increases for unsold goods,
(iii) there is no allocation, and (iv) the active bidders for each good are updated w.r.t.  their bids 
(Rules \ref{sa:price}-\ref{sa:bid}). 
The allocations are assigned in terminal states w.r.t. 
trades (Rule \ref{sa:alloctrade}).  The payment for an agent is the cumulative value of the selling price for her traded goods (Rule \ref{sa:payment}). 
Finally, after a terminal state, a numerical variable cannot change (Rule \ref{sa:loop}). 
Let $\Sigma_{sa}$ be the set of Rules \ref{sa:initial}-\ref{sa:loop}.

\paragraph{Representing as a model} 
Next, we address the model representation of the Simultaneous Ascending Auction (SAA). Let $\mathscr{M}_{sa}$ be the set of ST-models $M_{sa}$ defined for any $\mathcal{A}_{sa} \subseteq \mathcal{L}_{TBBLf}$, $I_{sa} \subset \mathbb{Z}$, and the constant values   $\mathsf{inc}, \mathsf{start} \in I_{sa, +}$ and $\mathsf{m}, \mathsf{n} \in I_{sa, +}\setminus \{0\}$. Let $\mathsf{max}_{sa}$ denote the largest value in $I_{sa}$.
Each $M_{sa}$ is defined as follows: 

\begin{itemize}
    \item $W_{sa} = \{\langle b_{1,1},\dots, b_{\mathsf{n,m}}, t_{1, 1}, \dots, t_{\mathsf{n,m}}, p, p_1, \dots, $ $ p_\mathsf{m} \rangle:  b_{i,j}, t_{i,j} \in \{0,1\} $ \& $ p, p_j \in I_{sa, +}$ \& $ i\in N_{sa}$ \& $ j\in G_{sa}\}$; 
    
    \item  $\bar{w}_{sa} = \langle 0,\dots, 0, 0,\dots, 0, \mathsf{start}, \mathsf{start}, \dots, \mathsf{start} \rangle$;

    \item $T_{sa}= \{w:  w = \langle   b_{1,1},\dots, b_{\mathsf{n,m}}, t_{1, 1}, \dots, t_{\mathsf{n,m}}$, $ p, p_1$, $\dots,p_\mathsf{m} \rangle \in W_{sa}\setminus \{\bar{w}_{sa}\}$ \& for all $j\in G_{sa}$, either (i) $t_{i,j} =1$ for some $i \in N_{sa}$ or (ii) $bid_{i,j} = 0$, for all $i\in N_{sa}\}$; 
    
    \item $L_{sa}=  \{(w, i, OR((\langle 1, 1, pr_1 \rangle$,  $\dots, \langle 1, 1, pr_\mathsf{m} \rangle), 0)):$ 
    $i \in N_{sa}$ \& 
    $w = \langle  b_{1,1},$ $\dots, $ $b_{\mathsf{n,m}}, t_{1, 1}, $ $\dots, $ $ t_{\mathsf{n,m}}, p, p_1,$ $\dots,$  $p_\mathsf{m} \rangle \in W_{sa}$ \&  for all $j \in G_{sa}$, and all $0\leq pr_j< \mathsf{max}_{sa}-\mathsf{inc}$ such that either (i) $pr_j = 0 $ \& $t_{i,j}=0$ or (ii) $ pr_j =p $ \& $t_{i',j} \neq 1$, for all $i' \in N_{sa}$ or (iii) $pr_j = p_j $ \& $ t_{i,j} =1\}$;
    
    \item $U_{sa}$ is defined as: for all $w = \langle  b_{1,1},\dots, b_{\mathsf{n,m}}, t_{1,1}$, $\dots$, $t_{\mathsf{n,m}}, p, p_1$, $\dots, p_\mathsf{m} \rangle \in W_{sa}$ and all $d \in \mathcal{A}^\mathsf{n}_{sa}$: 
    \begin{itemize}
        \item If $w \not \in T_{sa}$, then: 
        $U_{sa}(w, d) = \langle b_{1,1}'$, $\dots, b_{\mathsf{n,m}}', t_{1,1}',\dots$, $t_{\mathsf{n,m}}'$, $p', p_1',\dots, p_\mathsf{m}' \rangle $, where for every $i\in N_{sa}$ and $j\in G_{sa}$, 
        (i) $b_{i,j}' = 1$ iff 
        $d(i) =  OR((\langle 1, 1, pr_1 \rangle$,  $\dots, \langle 1, 1, pr_\mathsf{m} \rangle), 0)$ and $pr_j \neq 0$; and $b_{i,j}' = 0$ otherwise; 
        (ii) $t_{i,j}' = 1$ iff $b_{i,j}' = 1$ and for all $r \in N_{sa}\setminus\{i\}, b_{i,j}' \neq 1$; and $t_{i,j}' = 0$ otherwise;
        (iii) $p' = p+\mathsf{inc}$;
        (iv)  $p_j' = p_j+\mathsf{inc}$ iff $t_{r,j}' = 0$ for all $r \in N_{sa}$; and $p_j' = p_j$ otherwise.
        \item Otherwise, $U_{sa}(w,d) = w$.
    \end{itemize} 

    \item For each $w \in W_{sa}$, $ i \in N_{sa}$ and $j\in G_{sa}$, 
    (i) $\pi_{Y, sa}(w, trade_{i,j}) = t_{i,j}$;
    (ii) $\pi_{Y, sa}(w, price) = p$; 
    (iii) $\pi_{Y, sa}(w, price_{j}) = p_j$;
    (iv) $\pi_{Y, sa}(w, alloc_{i,j}) = 0$ iff $w\not \in T_{sa}$, and $\pi_{Y, sa}(w$, $alloc_{i,j}) = t_{i,j}$ otherwise; 
    (v) $\pi_{Y, sa}(w, payment_{i}) = \sum_{j \in G_{sa}} (p_j \pi_{Y, sa}(w$, $trade_{i,j}))$.
    
    \item For each $w \in W_{sa}$, $\pi_{\Phi, sa}(w) = \{sold_j : t_{i,j} =1 $ \& $ j\in G_{sa}$ \& $ i\in N_{sa}\} \cup \{ bid_{i,j} : b_{i,j} =1 $ \& $ j\in G_{sa} $ \& $ i\in N_{sa}\}$.
    
\end{itemize}

Hereafter, we assume an instance of $M_{sa} \in \mathscr{M}_{sa}$ and $\Sigma_{sa}$ for some 
$\mathcal{A}_{sa} \subseteq \mathcal{L}_{TBBLf}$, $I_{sa} \subset \mathbb{Z}$, $\mathsf{start} \in I_{sa, +}$ and $\mathsf{inc}, \mathsf{m}, \mathsf{n} \in I_{sa, +}\setminus \{0\}$. 

\begin{example}
Let $M_{sa} \in \mathscr{M}_{sa}$, where $\mathsf{start} = 2$, $\mathsf{inc} = 1$ and the agents and good sets are the same from Example \ref{example:ce}. 
Figure \ref{fig:saflow} illustrates a path in $M_{sa}$.
In state $w_0$, agents $r1$ and $r2$ bid for good $\mathsf{a}$, but only agent $r1$ bid for good $\mathsf{b}$. In state $w_1$, since $r1$ is the only bidder for $\mathsf{b}$, $\mathsf{b}$ is sold to her. 
Agent $r1$ needs to keep her bid for $\mathsf{b}$ and $r2$ can no longer bid for it. In $w_1$, agent $r2$ increases its bid for  good $\mathsf{a}$ and agent $r1$ do not bid for $\mathsf{a}$. In state $w_2$, since $r2$ is the only bidder for $\mathsf{a}$, she buys the good. Since all the goods were sold, this state is terminal.  

\begin{figure}[ht]
\centering
\includegraphics[width=1\columnwidth]{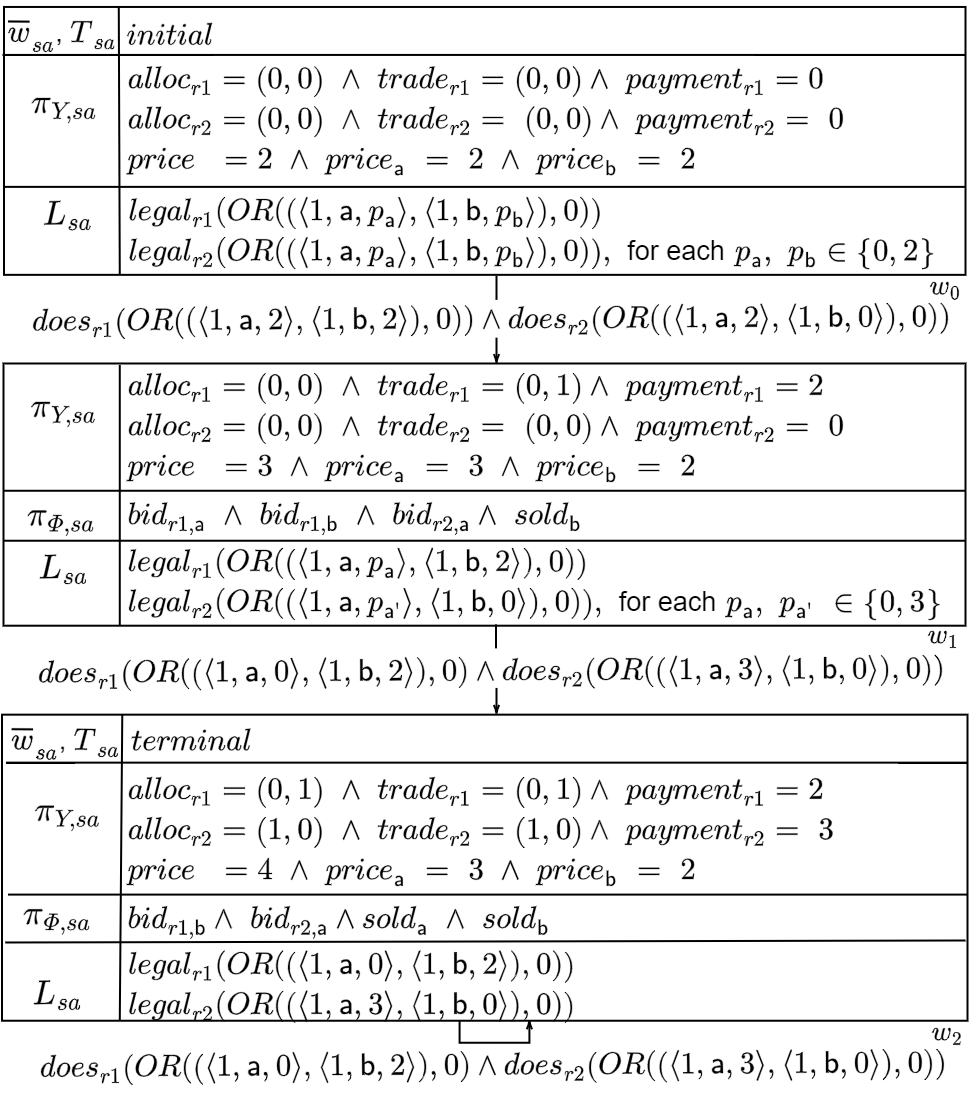}
\caption{A Path in $M_{sa}$, with 2 bidders and $2$ goods} \label{fig:saflow}
\Description{}
\end{figure}

\end{example}

Let us now evaluate the protocol. First, we show that  $\Sigma_{sa}$ is a sound representation of $M_{sa}$. 

\begin{proposition}
\label{prop:modelsa}
$M_{sa}$ is an ST-model and it is a model of  $\Sigma_{sa}$.
\end{proposition}
\begin{proof}\emph{(Sketch)} 
It is routine to check that $M_{sa}$ is actually an ST-model. Given a path $\delta$, any stage $t$ of $\delta$ in $M_{sa}$, we need to show that $M_{sa},\delta, t \models \varphi$, for each $\varphi \in \Sigma_{sa}$.
Let us verify Rule \ref{sa:initial}. Assume $M_{sa}, \delta, t \models initial$ iff $\delta[t] = \bar{w}_{sa}$. 
By the definition of $\bar{w}_{sa}$,  $\pi_{\Phi,sa}$ and $\pi_{Y,sa}$, we have $\pi_{Y, sa}(\bar{w}_{sa}, price) = \mathsf{start}$, $\pi_{Y, sa}(\bar{w}_{sa}, price_j) = \mathsf{start}$, $bid_{i,j} \not \in \pi_{\Phi, sa}(\bar{w}_{sa})$ and $trade_{i,j} = 0$, for all $i\in N_{sa}$ and $j\in G_{sa}$. Thus, $M_{sa}, \delta, t \models initial$ iff $M_{sa}, \delta, t \models price = \mathsf{start} $ $ \land  \bigwedge_{j\in G_{sa}} ( price_j = $   $\mathsf{start} \land \bigwedge_{i \in N} (\neg bid_{i,j} \land trade_{i,j} = 0))$. 
  
Now we verify Rule \ref{sa:sold}. Let $j \in G_{sa}$ be a good type. Assume $M_{sa}, \delta, t \models sold_j$ iff $sold_j \in \pi_{\Phi, sa}(\delta[t])$ iff $\pi_{Y, sa}(\delta[t], trade_{i,j}) = 1$ for some $j \in G_{sa}$ iff $M_{sa}, \delta, t \models \bigvee_{i \in N_{sa}} trade_{i,j} = 1$.

Next, we consider Rule $\ref{sa:terminal}$. Assume $M_{sa}, \delta, t \models terminal$ iff $\delta[t] \neq \bar{w}_{sa}$ and for all $j\in G_{sa}$, either $M_{sa}, \delta, t \models trade_{r,j} = 1$ for some $r \in N_{sa}$ or $M_{sa}, \delta, t \models \neg bid_{i,j}$ for all $i \in N_{sa}$. By Rule \ref{sa:sold}, $M_{sa}, \delta, t \models terminal$ iff $M_{sa}, \delta, t \models \neg initial \land \bigwedge_{j \in G_{sa}} (sold_j \lor \bigwedge_{j\in G_{sa}} \neg bid_{i,j})$.

Now we verify Rule \ref{sa:bid}. Let $i \in N_{sa}$ and $j\in G_{sa}$. Assume $M_{sa}, \delta, t \models (does_i(OR((\langle 1,1,p_1\rangle, \dots, \langle 1,\mathsf{m},p_\mathsf{m}\rangle), 0)) \land p_j \neq 0) \lor (bid_{i,j} \land terminal)$, for some $p_1, \dots, p_\mathsf{m} \in I_{sa, +}$. We next prove for the two cases. First, assume $M_{sa}, \delta, t \models bid_{i,j} \land terminal$. Then $bid_{i,j} \in \pi_{\Phi, sa}(\delta[t])$ and $\delta[t] \in T_{sa}$.
By the update function, $\delta[t+1] = \delta[t]$ and $M_{sa}, \delta, t+1 \models bid_{i,j}$, i.e., $M_{sa}, \delta, t \models \bigcirc bid_{i,j}$.
In the second case, assume $does_i(OR((\langle 1,1,p_1\rangle, \dots, \langle 1,\mathsf{m},p_\mathsf{m}\rangle), 0)) \land p_j \neq 0$. By the update function, $bid_{i,j}\in \pi_{\Phi, sa}(\delta[t+1])$ and thus $M_{sa}, \delta, t \models \bigcirc bid_{i,j}$.

The remaining rules are verified in a similar way.
\end{proof}

Next, we show that no good can be bought by  two different agents, \emph{i.e}, given any two agents and a good, one of them will have her  trade equal zero. 

\begin{proposition} \label{prop:sa}
For each $j \in G_{sa}$ and each $i, i' \in N_{sa}$ such that $i\neq i'$,
$M_{sa} \models trade_{i,j} = 0 \lor trade_{i',j} = 0$.
\end{proposition}
\begin{proof}
Given a path $\delta$ in $M_{sa}$, any stage $t$ of $\delta$ and a good type $j \in G_{sa}$, let $i, i' \in N_{sa}$, such that $i \neq i'$. 
If $\delta[t] = \bar{w}_{sa}$, then $M_{sa}, \delta, t \models trade_{i,j} =0 \land trade_{i',j} =0$ (see Rule \ref{sa:initial}). Otherwise, by the path definition, $\delta[j] = U_{sa}(\delta[t-1], \theta(\delta, t-1))$.
Let us suppose for the sake of contradiction that  $M_{sa}, \delta, t \not \models trade_{i,j} = 0 \lor trade_{i',j} = 0$. Since $W_{sa}$ construction defines $trade_{i,j}, trade_{i',j} \in \{0,1\}$, we have $M_{sa}, \delta, t \models trade_{i,j} = 1 \land trade_{i',j} = 1$. Thus, $M_{sa}, \delta, t-1 \models \bigcirc trade_{i,j}$. By Rule \ref{sa:trade1}, $M_{sa}, \delta, t-1 \models \bigcirc(bid_{i,j} \land \bigwedge_{r \in N_{sa}\setminus \{i\}} \neg bid_{r,j})$. Thereby, $M_{sa}, \delta, t-1 \not \models \bigcirc (bid_{i',j} \land \bigwedge_{r \in N_{sa}\setminus \{i'\}} \neg bid_{r,j})$ and $M_{sa}, \delta$, $t-1 \not \models \bigcirc trade_{i',j =1}$. Thus, $M_{sa}, \delta, t \not \models trade_{i',j} = 1$, which is a contradiction..
\end{proof}

The legal bids in a state respect the $buyer$ and $unit$ restrictions. It means that  agents cannot bid for negative prices and can only ask for one unit of each good.
\begin{proposition} \label{prop:salegal}
For each $i \in N_{sa}$ and each $a \in \mathcal{A}_{sa}$,
$M_{sa} \models legal_i(a) \rightarrow rest_i(buyer, a) \land rest_i(unit, a)$.
\end{proposition}
\begin{proof}
Let $\delta$ be a path in $M_{sa}$ and $t$ a stage of $\delta$. Assume $M_{sa} \models legal_i(a)$ iff $a \in L_{sa}(\delta[t],i)$. By $L_{sa}$ construction, $a = OR((\langle 1, 1, p_1 \rangle$, $\dots$, $\langle 1, \mathsf{m}, p_\mathsf{m} \rangle ), 0)$, for $ p_1, \dots, p_\mathsf{m} \in I_{sa, +}$.  
The set of leaves and nodes of $a$ are
$Leaf(a) = \{\langle 1, 1, p_1 \rangle, \dots,\langle 1, \mathsf{m}$, $p_\mathsf{m} \rangle\}$ and  $Node(a) = \{OR((\langle 1, 1, p_1 \rangle$, $\dots$, $\langle 1, \mathsf{m}, p_\mathsf{m} \rangle ), 0)\} \cup Leaf(a)$, resp. Thus, for all $\beta \in Node(a)$, $v_i(\beta) \geq 0$ and for all $\beta' \in Leaf(a)$, $q_i(\beta', j) = 1$ for some $j \in G_{sa}$. Thereby, $ M_{sa}, \delta, t \models rest_i(buyer, a) \land rest_i(unit, a)$.
\end{proof}

The next proposition shows $M_{sa}$ satisfies playability, that is, there is always a legal action for each agent to take.

\begin{proposition}
\label{prop:playabilitysa}
For each agent $i \in N_{ce}$,
$M_{sa} \models \bigvee_{a \in \mathcal{A}_{ce}}legal_i(a)$.
\end{proposition}
\begin{proof}
Given a path $\delta$ in $M_{sa}$ and a stage $t$ in $\delta$, we show that there is a legal action for agent $i$ in $\delta[t]$. For each $j \in G_{sa}$, let $p_j = 0$ if $\pi_{Y, sa}(\delta[t], trade_{i,j}) = 1$. Otherwise, let $p_j = \pi_{Y, sa}(\delta[t], price_j)$. By the definition of $L_{sa}$, we have $OR((\langle 1,1,p_1\rangle, \dots,\langle1$, $\mathsf{m},p_\mathsf{m}\rangle),0) \in L_{sa}(\delta[t], \allowbreak i)$. Thus, $M_{sa}, \delta, t \models \bigvee_{a \in\mathcal{A}_{sa}}legal_i(a)$.
\end{proof}

Each path in $M_{sa}$ reaches a terminal state, and thus the protocol satisfies the termination condition from GGP.

\begin{proposition}
\label{prop:terminationsa}
    For each path $\delta$ in $M_{sa}$, $\delta$ is  complete.
\end{proposition}
\begin{proof} 
Remind $\mathsf{start} \in I_{sa, +}$ and $\mathsf{inc}\in I_{sa, +}\setminus \{0\}$. 
Let $\delta$ be a path in $M_{sa}$. In $\delta[0]$, $\pi_{Y,sa}(\delta[0], price) = \mathsf{start}$. By the update function, for any stage $t$, if $\delta[t] \not \in T_{sa}$, then $\pi_{Y,sa}(\delta[t+1], price) = \pi_{Y,sa}(\delta[t], price)+\mathsf{inc}$.

For the sake of contradiction, let us assume $\delta$ is not complete. 
Let $i \in N_{sa}$ be any agent. 
By the definition of $L_{sa}$,  $legal_i(OR((\langle 1,1,p_1\rangle$, $\dots, \langle 1,\mathsf{m}, p_\mathsf{m}\rangle), 0)) \in L(\delta[t], i)$, 
for all $0 \leq p_j < \mathsf{max}_{sa}-\mathsf{inc}$ and $j\in G_{sa}$, such that either (i) $p_j = 0$ \& $\pi_{Y, sa}(\delta[t]$, $trade_{i,j}) = 0$, or (ii) $p_j = price$ \& $\pi_{Y, sa}(\delta[t], trade_{r,j}) = 0$ for all $r\in N_{sa}$, or (iii) $p_j = price_j$ \& $\pi_{Y, sa}(\delta[t], trade_{r,j}) = 0$. 
 Since $\pi_{Y,sa}(\delta[t+1], price)>\pi_{Y,sa}(\delta[t], price)$, there will be a stage $e\geq0$ in $\delta$, where the condition (ii) will not be true for any $0\leq p_j < \mathsf{max}_{sa}-\mathsf{inc}$.
 
 Thus, for each good $j$, it will be the case that $i$ bids $0$ for it or the good was assigned to her (i.e., $\pi_{Y,sa}(\delta[e], trade_{i,j})= 1$). From Rules \ref{sa:terminal} and \ref{sa:bid} in $\Sigma_{sa}$, it follows that $\delta[e+1]\in T_{sa}$. Thus, $\delta$ is a complete path, which is a contradiction. 
\end{proof}

From being a single-side auction where all agents are buyers, it follows that there is no-deficit in $M_{sa}$, but it is not budget-balanced.
\begin{proposition}
\label{prop:saBBnoDef}
$M_{sa} \not \models BB$ and $M_{sa}  \models noDeficit$. 
\end{proposition}
\begin{proof} \emph{(Sketch)}  
Given a path $\delta$ in $M_{sa}$ and a stage $t$ in $\delta$, let us show a counterexample.
Note that each allocation can be either 0 or 1 and the good price is at least 0, i.e., $\pi_{Y, sa}(\delta[t], alloc_{i,j}) \{0, 1\}$ and $\pi_{Y, sa}(\delta[t], price_j) \in I_{sa, +}$.
Assume $\pi_{Y, sa}(\delta[t], alloc_{i',j'}) = 1$ and $\pi_{Y, sa}(\delta[t], price_{j'}) >0$, for some $j'\in G_{sa}$ and some $i' \in N_{sa}$. It follows from Rule \ref{sa:payment} that $M_{sa}, \delta, t \models payment_{i'} >0$ and  $M_{sa}, \delta, t \models add(payment_1, \dots, payment_\mathsf{n}) >0$. Thus, 
$M_{sa} \not \models BB$ and $M_{sa}  \models noDeficit$. 
\end{proof}

Finally, the agents can always ensure that their utility will be at least as good in the next state as it was in the current, \emph{i.e},  $IR_i$ is valid in $M_{sa}$, for each $i$.
\begin{thm}
For each $i \in N_{sa}$ and some valuation $\vartheta_i$ over individual trades, 
 $M_{sa} \models IR_i$. 
\end{thm}
\begin{proof}
Given a path $\delta$ in $M_{sa}$ and a stage $t$, 
let $T_i = OR((\langle 1, 1, p_1\rangle$, $\dots, \langle 1, \mathsf{m}, p_\mathsf{m}\rangle), 0)$, where $p_j = 0$ if $\pi_{Y, sa}(\delta[t],trade_{i, j}) =0$; otherwise $p_j = \pi_{Y, sa}(\delta[t]$, $price_j)$, for each $j\in G_{sa}$.
Since $T_i \in L_{sa}(\delta[t], i)$, we can construct a path $\delta'$, such that $\delta'[0,t] = \delta[0,t]$, $\theta_i(\delta', t) = T_i$, and $\theta_r(\delta', t) = \theta_r(\delta, t)$, for all $r\in N_{sa}\setminus \{i\}$.
Thus, $M_{ce}, \delta', t \models utility_i(trade_i$, $payment_i) = x \rightarrow \bigcirc  utility_i(trade_i$, $payment_i) = x$, for $x\in I_{sa}$ and $M_{ce}, \delta,t \models IR_r$. 
\end{proof}

\section{Discussion and Conclusion}
\balance
In this paper, we have presented a unified framework for representing auction protocols. Our work is at the frontier of auction theory and knowledge representation.

\paragraph{Related work}

Our work is rooted in the key contributions on Combinatorial Auctions \cite{Nisan2000,Nisan2004,Xia2005,Parkes2005}. All these works adopt a mechanism design perspective: they focus on the properties of a given protocol and bidding language. Our work has a different purpose. The CEDL language includes the bidding part of a protocol, but also the protocol itself. Such a language can be used to automatically derive properties for these protocols. CEDL can also be used as a framework for testing new auction protocols. 

To the best of our knowledge, almost all contributions on the computational representation of 
auction-based markets focus on the implementation of the winner determination problem. For instance, Baral and Ulyan (\citeyear{BU01}) show how a specific auction, namely combinatorial auctions, 
can be encoded in a logic program. 
A hybrid approach mixing linear programming and logic programming has been proposed by Lee and Lee (\citeyear{LL97}): they focus on sealed-bid auctions and show how qualitative reasoning helps to refine the optimal quantitative solutions. 
Giovannucci et al.  (\citeyear{Giovannucci2010}) explore a graphical formalism to compactly represent the winner determination problem for Multi-Unit Combinatorial Auctions.
The closest contributions to ours are the \emph{Market Specification Language} (MSL) \cite{TZ10} and ADL \cite{MP2020}, also based on GDL. They both focus on representing single good auction through a set of rules and then interpreting an auction-instance with the help of a state-based semantics. MSL is limited to single agent perspective while ADL is not. However, the main limit of both approaches is the absence of a bidding language.

\paragraph{Going Further}
First direction is Computational complexity. Although the model-checking (MC) problem in ADL
is PTIME \cite{MP2020}, 
the winner determination in Combinatorial Auctions (and thus also in Combinatorial Exchange) is known to be NP-complete \cite{rothkopf1998computationally}. We aim to explore how the MC problem in CEDL is affected by these results. Clearly, the fragment of CEDL without formulas referring to the WD mixed-integer program is still PTIME. For instance, the Simultaneous Ascending Auction protocol described in this paper avoids such formulas.

CEDL definitely puts the emphasis on the auctioneer and auction designer. Our second direction is to design a CEDL-based \textit{General Auction Player} (GAP) that can interpret and reason about the rules of an auction-based market. The key difference, when the players' perspective is considered, is the epistemic and strategic aspects: players have to reason about other players' behavior.
The epistemic component will allow an agent to bid according to her beliefs about other agents' private values. Our future GAP should then be based on the epistemic extensions of GDL such as GDL-III \cite{Thielscher2016} and Epistemic GDL \cite{Jiang2016a}. 



  

\begin{acks}
This research is supported by the ANR project AGAPE ANR-18-CE23-0013. 
\end{acks}



\bibliographystyle{ACM-Reference-Format} 


\end{document}